\documentclass[acmsmall,screen,nonacm]{acmart} 

\usepackage{multirow}
\usepackage{microtype}
\usepackage{listings}
\usepackage{enumitem}
\setlist{itemsep=3pt}
\usepackage[noabbrev,nameinlink]{cleveref}
\usepackage{xfrac}
\usepackage{fancybox}
\usepackage{keyval}
\usepackage{zref-savepos,zref-user}
\usepackage{hyperref}

\DeclareMathOperator{\bO}{\mathsf{O}}
\newcommand*{\size}[1]{\lvert {#1} \rvert}
\newcommand*{\raml}{\textsf{RaML}}
\newcommand*{\atlas}{\ensuremath{\mathsf{ATLAS}}}
\newcommand*{\tachis}{\ensuremath{\mathsf{Tachis}}}
\newcommand*{\coq}{\ensuremath{\mathsf{Rocq}}}
\newcommand*{\isabelle}{\ensuremath{\mathsf{Isabelle/HOL}}}
\newcommand*{\easycrypt}{\ensuremath{\mathsf{EasyCrypt}}}

\newcommand{\denotOf}[1]{\llbracket{#1}\rrbracket}
\newcommand{\remRef}[1]{\lfloor{#1}\rfloor}
\newcommand{\defsymb}{::=}
\newcommand{\evalContextSymb}{\mathbb{C}}
\newcommand{\evalContext}[1]{\evalContextSymb[#1]}
\newcommand{\lcName}{\lambda^\mathit{PC}}
\newcommand{\smallStep}{\hookrightarrow}

\def\LH/{Liquid Haskell}

\NewDocumentCommand{\DEF}{m}{\detokenize{#1}\label{lst:\detokenize{#1}}}
\NewDocumentCommand\FN{ O{#2} m O{black} D<>{\BooleanTrue} }{\relax\ifmmode\hyperref[lst:#1]{\color{#3}\mathtt{#2}}\else\hyperref[lst:#1]{\color{#3}\IfBooleanT{#4}{\footnotesize}\texttt{#2}}\fi}

\def\listIn/{\FN[in]{\in}}

\newcommand{\secRef}[1]{\hyperref[#1]{\S~\ref*{#1}}}

\def\thmSumPermute/{\FN[thm_sum_permute]{thm\_sum\_permute}}
\def\thmSumSplit/{\FN[thm_sum_split]{thm\_sum\_split}}

\ifdefined\WithAppendix
    \def\thmSumFactor/{\FN[thm_sum_factor]{thm\_sum\_factor}}    
    \def\thmSumLinear/{\FN[thm_sum_linear]{thm\_sum\_linear}}
    \def\thmSumReverse/{\FN[thm_sum_reverse]{thm\_sum\_reverse}}
    \def\thmSumConstant/{\FN[thm_sum_constant]{thm\_sum\_constant}}
    \def\thmSumRewrite/{\FN[thm_sum_rewrite]{thm\_sum\_rewrite}}
\else
    \def\thmSumFactor/{{\footnotesize\texttt{thm\_sum\_factor}}}
    \def\thmSumLinear/{{\footnotesize\texttt{thm\_sum\_linear}}}
    \def\thmSumReverse/{{\footnotesize\texttt{thm\_sum\_reverse}}}
    \def\thmSumConstant/{{\footnotesize\texttt{thm\_sum\_constant}}}
    \def\thmSumRewrite/{{\footnotesize\texttt{thm\_sum\_rewrite}}}
\fi

\ifdefined\InAppendix
    \def\Outcome/{{\footnotesize\color{constructor}\texttt{Outcome}}}
    \def\DEmpty/{{\footnotesize\color{constructor}\texttt{DEmpty}}}
    \def\DSome/{{\footnotesize\color{constructor}\texttt{DSome}}}
    \def\ProperDist/{{\footnotesize\color{constructor}\texttt{ProperDist}}}
    \def\SubDist/{{\footnotesize\color{constructor}\texttt{SubDist}}}
    \def\pure/{{\footnotesize\texttt{pure}}}
    \def\dunit/{{\footnotesize\texttt{dunit}}}
    \def\expectCost/{{\footnotesize\texttt{expectCost}}}
    \def\expectVal/{{\footnotesize\texttt{expectVal}}}
    \def\append/{{\footnotesize\texttt{append}}}
    \def\support/{{\footnotesize\texttt{support}}}
    \def\probOf/{{\footnotesize\texttt{probOf}}}
    \def\fMap/{{\footnotesize\texttt{fMap}}}
    \def\pMap/{{\footnotesize\texttt{pMap}}}
    \def\tick/{{\footnotesize\texttt{tick}}}
    \def\probMass/{{\footnotesize\texttt{probMass}}}
    \def\combine/{{\footnotesize\texttt{combine}}}
    \def\bind/{{\footnotesize\texttt{bind}}}
    \def\coin/{{\footnotesize\texttt{coin}}}
    \def\uniform/{{\footnotesize\texttt{uniform}}}
    \def\finSum/{{\footnotesize\texttt{finSum}}}
    \def\noDuplicates/{{\footnotesize\texttt{noDuplicates}}}

    \def\permute/{{\footnotesize\texttt{permute}}} 
    \def\permuteBody/{{\footnotesize\texttt{permuteBody}}}

    \def\logTwo/{{\footnotesize\texttt{log2}}} 
    \def\logTwoGeZero/{{\footnotesize\texttt{thm\_log2\_ge\_0}}}
    \def\logTwoIneq/{{\footnotesize\texttt{thm\_log2\_inequality}}}

    \NewDocumentCommand{\HfinSum}{ m m }{\color{black}\finSumSymb{#1}{#2}}
    \NewDocumentCommand\HEval{m m}{\color{black}\Eval{#1}{#2}}
    \NewDocumentCommand\HEvalTT{m m}{\color{black}\Eval{\mathtt{#1}}{\mathtt{#2}}}
    
    \NewDocumentCommand\HEvalTTOnly{ m }{\color{black}\mathbb{E}_{\mathtt{#1}}}
    
    \NewDocumentCommand\HEcostOpen{}{\EcostOpen}

    \def\triRight/{\color{black}\triangleRight}
    \def\Range/{{\footnotesize\texttt{Range}}} 
    \NewDocumentCommand\Rnneg{}{\mathbb{R}^{\geqslant 0}}
\else
    \def\List/{\FN{List}[constructor]}
    \def\Nil/{\FN[List]{Nil}[constructor]}
    \def\Cons/{\FN[List]{Cons}[constructor]}
    \def\Heap/{\FN{Heap}[constructor]}
    \def\Empty/{\FN[Heap]{Empty}[constructor]}
    \def\Outcome/{\FN{Outcome}[constructor]}
    \def\DEmpty/{\FN[data-dist]{DEmpty}[constructor]}
    \def\DSome/{\FN[data-dist]{DSome}[constructor]}
    \def\pure/{\FN{pure}}
    \def\dunit/{\FN{dunit}}
    \def\expectCost/{\FN{expectCost}}
    \def\expectVal/{\FN{expectVal}}
    \def\append/{\FN{append}}
    \def\support/{\FN{support}}
    \def\probOf/{\FN{probOf}}
    \def\fMap/{\FN{fMap}}
    \def\pMap/{\FN{pMap}}
    \def\tick/{\FN{tick}}
    \def\probMass/{\FN{probMass}}
    \def\combine/{\FN{combine}}
    \def\bind/{\FN{bind}}
    \def\coin/{\FN{coin}}
    \def\ProperDist/{\FN{ProperDist}[constructor]}
    \def\SubDist/{\FN{SubDist}[constructor]}
    \def\uniform/{\FN{uniform}}
    \def\finSum/{\FN{finSum}}
    \def\noDuplicates/{{\footnotesize\texttt{noDuplicates}}}

    \def\permute/{\FN{permute}}
    \def\permuteBody/{\FN{permuteBody}}
    
    \def\countHires/{\FN{countHires}}
    \def\thmExpectValHireAssistant/{\FN[thm_expectVal_hireAssistant]{thm\_expectVal\_hireAssistant}}

    \def\Event/{\FN{Event}[constructor]}
    \def\evProb/{\FN{evProb}}
    \def\andF/{\FN{andF}}
    \def\condProb/{\FN{condProb}}
    \def\thmEvProbExt/{\FN[thm_evProb_ext]{thm\_evProb\_ext}}
    \def\thmJointProbZero/{\FN[thm_jointProb_zero]{thm\_jointProb\_zero}}

    \def\logTwo/{\FN[logTwo]{log\textsubscript{2}}}
    \def\logTwoBase/{\FN[log2_base]{log\textsubscript{2}\_base}}
    \def\logTwoProduct/{\FN[log2_product]{log\textsubscript{2}\_product}}
    \def\logTwoMonotone/{\FN[log2_monotone]{log\textsubscript{2}\_monotone}}
    \def\logTwoGeZero/{\FN[thm_log2_ge_0]{thm\_log\textsubscript{2}\_ge\_0}}
    \def\logTwoIneq/{\FN[thm_log2_inequality]{thm\_log\textsubscript{2}\_inequality}}

    \NewDocumentCommand{\HfinSum}{ m m }{\hyperref[lst:finSum]{\color{black}\finSumSymb{#1}{#2}}}
    \NewDocumentCommand\HEval{m m}{\hyperref[lst:expectVal]{\color{black}\mathbb{E}}_{#1}[#2]}
    \NewDocumentCommand\HEvalTT{m m}{\hyperref[lst:expectVal]{\color{black}\mathbb{E}}_{\mathtt{#1}}[\mathtt{#2}]}
    
    \NewDocumentCommand\HEvalTTOnly{ m }{\hyperref[lst:expectVal]{\color{black}\mathbb{E}}_{\mathtt{#1}}}
    
    \NewDocumentCommand\HEcostOpen{}{\hyperref[lst:expectCost]{\color{black}\EcostOpen}}

    \def\triRight/{\hyperref[lst:triangle]{\color{black}\triangleRight}}
    \def\Range/{\FN{Range}}
    \NewDocumentCommand\Rnneg{}{\hyperref[lst:rnneg]{\color{black}\mathbb{R}^{\geqslant 0}}}
    \NewDocumentCommand\Zpos{}{\hyperref[lst:zpos]{\color{black}\mathbb{Z}^+}}
\fi

\def\rquick/{\FN{rquick}}
\def\rquickBody/{\FN{rquickBody}}
\def\partition/{\FN{partition}}
\def\sorted/{\C!sorted!}
\def\elems/{\C!elems!}
\def\deleteAt/{\C!deleteAt!} 
\def\factorial/{{\footnotesize\texttt{factorial}}} 
\def\harmonic/{{\footnotesize\texttt{harmonic}}} 

\def\qselect/{\FN{qselect}}
\def\qselectBody/{\FN{qselectBody}}

\def\thmProbOfConsNeq/{\FN[thm_probOf_cons_neq]{thm\_probOf\_cons\_neq}} 
\def\thmProbOfConsEq/{\FN[thm_probOf_cons_eq]{thm\_probOf\_cons\_eq}}

\def\meld/{\FN{meld}}

\NewDocumentCommand\Ecost{m}{%
    \mathbb{E}^{\text{cost}}[#1]
}
\NewDocumentCommand\HEcost{m}{
    \hyperref[lst:expectCost]{\color{black}\Ecost{#1}}
}
\NewDocumentCommand\HEcostTT{m}{
    \hyperref[lst:expectCost]{\color{black}\Ecost{\mathtt{#1}}}
}
\NewDocumentCommand\EcostOnly{}{\mathbb{E}^{\text{cost}}}
\NewDocumentCommand\HEcostOnly{}{\hyperref[lst:expectCost]{\color{black}\EcostOnly}}

\NewDocumentCommand\EcostOpen{}{\mathbb{E}^{\text{cost}}[}

\newcommand{\finSumSymbOnly}{\hyperref[lst:finSum]{\color{black}\sum\hspace{-0.75em}\sum}}

\NewDocumentCommand\finSumSymbLimits{ m m }{%
    \sum\hspace{-1.75em}\sum_{#1}^{#2}
}
\NewDocumentCommand\finSumSymb{ m m }{%
    \sum\hspace{-0.75em}\sum\nolimits_{#1}^{#2}
}

\NewDocumentCommand{\HfinSumTT}{ m m }{%
    \HfinSum{\mathtt{#1}}{\mathtt{#2}}
}
\NewDocumentCommand{\HfinSumTTLimits}{ m m }{%
    \hyperref[lst:finSum]{\color{black}\finSumSymbLimits{\mathtt{#1}}{\mathtt{#2}}}
}

\def\collect/{\FN{collect}}
\def\collectAux/{\FN{collectAux}}
\def\geoFin/{\FN{geoFin}}

\newenvironment{tightcenter}{%
  \setlength\topsep{3pt}
  \setlength\parskip{3pt}
  \begin{center}
}{%
  \end{center}
}

\newcommand{\Rocq}{}
\def\Rocq/{\textsf{Rocq}}

\newsavebox\lstbox

\newcounter{hsExpr}[subsection]
\renewcommand\thehsExpr{\thesection.\arabic{hsExpr}}

\definecolor{cbViolet}{RGB}{231, 212, 232}
\definecolor{cbGreen}{RGB}{217, 240, 211}

\definecolor{lightGreen}{HTML}{f0f7f0}
\definecolor{darkerGreen}{HTML}{008000}

\newcommand{\greenBox}[1]{\fcolorbox{darkerGreen}{lightGreen}{#1}}

\newcommand{\verificationEffort}[1]{%
\noindent%
\cornersize{.2}%
\begingroup%
\setlength{\fboxsep}{5pt}%
\begin{Sbox}
\begin{minipage}{\dimexpr(\linewidth-15pt)}
\textbf{Verification Effort:} #1
\end{minipage}
\end{Sbox}
\Ovalbox{\TheSbox}%
\endgroup%
}

\newcommand{\labelExp}[1]{\hfill\refstepcounter{hsExpr}$\textcolor{black}{(\text{\normalfont\thehsExpr})}$\label{#1}}

\def\LOCquicksort/{466}
\def\LOCmeldableHeaps/{68}
\def\LOCpermutations/{77}
\def\LOCquickselect/{393}

\makeatletter
\define@key{resultRow}{name}{\gdef\resultRow@name{#1}}
\define@key{resultRow}{isabelle}{\gdef\resultRow@isabelle{#1}}
\define@key{resultRow}{tachis}{\gdef\resultRow@tachis{#1}}
\define@key{resultRow}{lh}{\gdef\resultRow@lh{#1}}
\define@key{resultRow}{rocq}{\gdef\resultRow@rocq{#1}}
\define@key{resultRow}{easycrypt}{\gdef\resultRow@easycrypt{#1}}

\newcommand{\resultRow}[1]{%
\gdef\resultRow@name{---}%
\gdef\resultRow@isabelle{---}%
\gdef\resultRow@tachis{---}%
\gdef\resultRow@lh{---}%
\gdef\resultRow@rocq{---}%
\gdef\resultRow@easycrypt{---}%
\setkeys{resultRow}{#1}%
\resultRow@name{} & \resultRow@easycrypt{} & \resultRow@isabelle{} & \resultRow@rocq{} & \resultRow@tachis{} & \resultRow@lh{} \\
}
\makeatother

\newcommand{\hsBasicStyle}{\color{black}\footnotesize\ttfamily}
\newcommand{\triangleRight}{\scalebox{1.2}{\rotatebox[origin=c]{-90}{$\bigtriangleup$}}}

\definecolor[named]{constructor}{HTML}{009304}
\definecolor[named]{commentGray}{HTML}{4f4f54}

\newcommand\hsCommentStyle{\rmfamily\itshape\color{commentGray}}

\lstdefinelanguage{haskell}{
    string = [b]{"},
    morecomment = [l]{--},
    commentstyle = \hsCommentStyle,
    stringstyle = \color{blue},
	keywordstyle = [1]\color{constructor},
	keywordstyle = [2]\bfseries,
	keywordstyle = [3]\bfseries\color{red},
	keywords = [1]{Empty,Heap,Dist,Nat,Double,Int,Outcome,List,Nil,Cons,ProperDist,SubDist,NoDupList,IncrList,String,Node,Tree,DEmpty,DSome,Event},
	keywords = [2]{data,type,Ord,Eq,reflect,measure,case,of,where,if,then,else,infix,let,in,import,lazy,QED},
	keywords = [3]{assume},
}

\newcommand\Nat{\mathbb{N}}
\newcommand\Int{\mathbb{Z}}
\newcommand\Bool{\mathbb{B}}
\newcommand\Real{\mathbb{R}}

\lstdefinestyle{haskell}{
    columns=fullflexible,
    language=haskell,
    aboveskip=0.75\medskipamount,
    belowskip=0.75\medskipamount,
    numberstyle=\tiny\color{black},
    basicstyle=\hsBasicStyle,
    showstringspaces = false,
    basewidth={.5em,0.4em},
    escapechar=~,
    mathescape,
	keepspaces,
    xleftmargin=2em,
    captionpos=b,
	literate=
	  {->}{$\to$}{2}
      {<-}{$\gets$}{2}
	  {|>}{\triRight/}{2}
	  {<=}{$\leqslant$}2
	  {>=}{$\geqslant$}2
	  {!=}{$\neq$}2
	  {=>}{$\Rightarrow$}2
      {<=>}{$\Leftrightarrow$}{2}
	  {\\}{$\lambda$}1
      {\#a}{$\alpha$}1
      {\#b}{$\beta$}1
      {\#c}{$\gamma$}1
      {\#N}{$\Nat$}1
      {\#Z}{$\Int$}1
      {\#B}{$\Bool$}1
      {\#R}{$\Real$}1
      {\#Rnneg}{$\Rnneg$}3
}

\NewDocumentCommand\C{ O{} }{\lstinline[style=haskell,basicstyle=\color{black}\footnotesize\ttfamily,#1]}

\begin{document}

\title{A Cost-Aware Probability Monad for Liquid Haskell}

\author{Matthias Hetzenberger}
\orcid{0000-0002-2052-8772}
\affiliation{%
  \institution{Vienna University of Technology}
  \department{Institute of Logic and Computation}
  \city{Vienna}
  \country{Austria}
}
\email{matthias.hetzenberger@tuwien.ac.at}

\author{Georg Moser}
\orcid{0000-0001-9240-6128}
\affiliation{%
  \institution{University of Innsbruck}
  \department{Department of Computer Science}
  \city{Innsbruck}
  \country{Austria}
}
\email{georg.moser@uibk.ac.at}

\author{Florian Zuleger}
\orcid{0000-0003-1468-8398}
\affiliation{%
  \institution{Vienna University of Technology}
  \department{Institute of Logic and Computation}
  \city{Vienna}
  \country{Austria}
}
\affiliation{%
  \institution{TU München}
  \department{TUM School of Computation, Information and Technology (CIT)}
  \city{Munich}
  \country{Germany}
}
\email{florian.zuleger@tuwien.ac.at}

\begin{abstract}
Probabilistic algorithms and data structures are widely used to obtain favourable expected performance guarantees.
While their mathematical analysis is often well understood, mechanising expected-cost analyses remains challenging, requiring reasoning about probability distributions, expectations, and recursive stochastic behaviour.
Existing formal approaches frequently require substantial manual proof effort, since expected costs are often encoded separately from probabilistic computations and must therefore be propagated explicitly throughout proofs.

In this paper, we present a cost-aware probability monad for \LH/ that supports reasoning about probabilistic programs together with their expected costs.
Our approach combines executable probabilistic programs with refinement-type-based verification and SMT-supported automation.
The monad intrinsically tracks probability mass, expected values, and expected costs through refinement types, enabling many quantitative properties of probabilistic computations to be inferred compositionally from program structure.

We evaluate our approach on several classical probabilistic algorithms and data structures, including meldable heaps, randomised quicksort and quickselect, randomised splay trees, random permutations, and the hiring problem.
The case studies demonstrate different points along the spectrum between automated and interactive verification.
\end{abstract}

\maketitle
\renewcommand{\shortauthors}{Matthias Hetzenberger, Georg Moser, Florian Zuleger}

\section{Introduction}

\begin{table}[t]
\caption{LoC Comparison}
\label{table:loc-comp}

\renewcommand{\arraystretch}{1.1}
\begin{tabular}{|l|r|r|r|r|r|}
\hline
Case Study             & \easycrypt{} & \isabelle{} & \coq{} & \tachis{} \cite{HaselwarterLMG024} & LH \\ \hline
\resultRow{name={Meldable Heaps}, tachis=1045, lh={\LOCmeldableHeaps/}}
\resultRow{name={Randomised Quicksort},isabelle={295 \cite{EberlHN20}}, tachis=1072, rocq={980 \cite{WeegenMcKinna}, 1179 \cite{Tassarotti018}}, lh={\LOCquicksort/}}
\resultRow{name={Randomised Quickselect}, easycrypt={370  \cite{AvanziniBGMV24}}, lh={\LOCquickselect/}}
\resultRow{name={Random Permutations}, isabelle={175 \cite{Fisher_Yates-AFP}}, lh={\LOCpermutations/}}
\hline
\end{tabular}

%
%

\end{table} 


\begin{table}[t]
\caption{LoC of \LH/ Formalisations. \\ Column \textbf{Code} states lines of executable Haskell code, \textbf{Annot.} states lines of \LH/ annotations,  \textbf{Proof} reports lines of Haskell code used for proving theorems, and column \textbf{All} are the overall lines of code.}
\label{table:loc}

\aboverulesep=0ex
\belowrulesep=0ex
\renewcommand{\arraystretch}{1.2}
\hspace{-0.5cm}%
\begin{tabular}{r|l|r|r|r|r|}
\cmidrule{2-6}
& \textbf{Subject} & \textbf{Code} & \textbf{Annot.}  & \textbf{Proof} & \textbf{All} \\
\cmidrule{2-6}
\multirow{3}{*}{\parbox{1.3cm}{\raggedleft Library\\Code}$\,\left\lbrace\begin{array}{l} \\ \\ \\ \end{array}\right.$\hspace{-0.45cm}}
& Cost-Aware Probability Monad (\secRef{sec:prob-monad}) & 234 & 283 & 112 & 629 \\
& Finite Sum Library (\secRef{sec:summations}) & 97 & 95 & 172 & 364 \\
& Logarithm Axiomatisation (\secRef{sec:logarithms}) & 12 & 10 & 16 & 38 \\
\multirow{7}{*}{\parbox{1.3cm}{\raggedleft Case\\Studies}$\,\left\lbrace\begin{array}{l} \\ \\ \\ \\ \\ \\ \\ \end{array}\right.$\hspace{-0.45cm}}
& Meldable Heaps (\secRef{sec:meldable-heaps}) & 30 & 22 & 0 & 52 \\
& Randomised Quicksort (\secRef{sec:rand-quick}) & 127 & 147 & 152 & 426 \\
& Randomised Splay Trees (\secRef{sec:splay-trees}) & 175 & 115 & 0 & 290 \\
& Random Permutations (\secRef{sec:random-perm}) & 25 & 24 & 39 & 88 \\
& Hiring Problem (\secRef{sec:hiring}) & 59 & 46 & 114 & 219 \\
& Bayes' Theorem (\secRef{sec:bayes}) & 13 & 12 & 10 & 35 \\
& Randomised Quickselect (\secRef{sec:qselect}) & 113 & 87 & 156 & 356 \\
\cmidrule{2-6}
\end{tabular}

\end{table}

The analysis of probabilistic programs has become increasingly prominent in programming languages and formal verification research (see e.g.~\cite{BKS:2020,NipkowB19,EberlHN20,LMZ:2022,VasilenkoVB22,AvanziniBGMV24}). Randomisation is widely used to obtain favourable expected performance guarantees and to simplify algorithmic design.
Classical examples include \emph{randomised quicksort}~\cite{Cormen:2009,cichon-quick}, \emph{randomised quickselect}~\cite{devroye1984exponential}, \emph{meldable heaps}~\cite{MartinezR98}, \emph{randomised splay trees}~\cite{ALBERS2002213}, \emph{random permutations}~\cite{durstenfeld1964algorithm}, and the \emph{hiring problem}~\cite{ajtai2001improved}.

While the mathematical analysis of such algorithms is often well understood, mechanising their expected-cost analyses remains challenging. In addition to functional correctness, one must reason about probability distributions, expectations, recursive stochastic behaviour, and accumulated computational cost. As a consequence, formal probabilistic analyses frequently require substantial proof engineering effort.

A central difficulty is that expected-cost reasoning is often not integrated directly into the representation of probabilistic computations themselves. In many existing approaches, probability distributions and expected costs are treated separately, requiring expectations and costs to be propagated explicitly throughout proofs. This complicates compositional reasoning and limits automation, since even relatively simple probabilistic analyses may require a considerable amount of auxiliary reasoning about how costs interact with monadic operations and recursive probabilistic computations.

In this paper, we address this problem by introducing a cost-aware probability monad for \LH/. Our key idea is to make expected costs intrinsic to probabilistic computations through refinement-typed monadic operators. The monad tracks probability mass, expected values, and expected costs directly at the type level. Consequently, many quantitative properties of probabilistic computations can be inferred compositionally from the structure of programs together with the refinement types of the monadic operators.

Our approach combines executable probabilistic programs with refinement-type-based verification and SMT-supported automation. Since the framework is embedded directly in \LH/, we focus on finite discrete probability distributions, which align naturally with refinement reflection, termination checking, and SMT-supported reasoning. Consequently, probabilistic algorithms remain ordinary executable Haskell programs rather than purely semantic encodings.
At the same time, refinement types allow us to express quantitative properties of probabilistic computations directly in types, while \LH/'s SMT-based verification infrastructure discharges many proof obligations automatically. More sophisticated mathematical arguments can furthermore be carried out using \LH/'s interactive theorem proving features. Thus, our framework supports both automated verification and interactive proofs within a uniform setting.

We evaluate our approach on several classical probabilistic algorithms and data structures, including meldable heaps, randomised quicksort and quickselect, randomised splay trees, random permutations, and the hiring problem.
These examples illustrate different points along the spectrum between highly automated and interactive verification.
While several expected-cost analyses can be verified with comparatively little user guidance, more sophisticated examples additionally demonstrate how the framework integrates smoothly with interactive theorem proving when deeper mathematical reasoning is required.
\Cref{table:loc-comp} compares the size of our formalisation with formalisations carried out using other theorem provers, showing that our approach is competitive.
\Cref{table:loc} summarises the size of our \LH/ developments, where we report the lines of executable Haskell code, the lines of \LH/ annotations, as well as the lines of mechanised proofs, implemented using the theorem proving capabilities of \LH/.
As indicated by \Cref{table:loc}, the developments evaluate the usability of the probability monad on the continuum between fully automated verification without proof obligations and interactive verification requiring the formalisation of proofs.
Overall, the case studies suggest that the proposed probability monad provides a lightweight and expressive foundation for mechanising expected-cost analyses of probabilistic programs in \LH/.

\paragraph{Our contributions}
We present a general-purpose library---encoded as a cost-aware probability monad---for implementing and verifying probabilistic algorithms. More precisely, we make the following contributions:

\begin{enumerate}

\item \emph{A cost-aware probability monad for \LH/.}
We introduce a refinement-typed probability monad for finite discrete probability distributions that intrinsically tracks probability mass, expected values, and expected costs, enabling compositional expected-cost reasoning directly at the type level.

\item \emph{Automated and interactive expected-cost verification.}
We demonstrate that the refinement types of the monadic operators enable a high degree of SMT-supported automation for probabilistic expected-cost analyses while integrating smoothly with \LH/'s interactive theorem proving features for mathematically more sophisticated arguments.

\item \emph{Soundness of the induced expected-cost analysis.}
We state and establish soundness of the expected-cost analysis induced by the introduced probability monad, cf.~\Cref{thm:soundness}.

\item \emph{Case studies.}
We evaluate the framework on several classical probabilistic algorithms and data structures, including meldable heaps, the hiring problem, randomised quicksort and quickselect, and randomised splay trees, illustrating different points along the spectrum between highly automated and interactive verification.
To this end, we discuss the verification effort of each case study.
\end{enumerate}

\paragraph{Outline}

In \Cref{sec:overview}, we motivate our \emph{cost-aware probability monad} for \LH/ and illustrate its use on the case studies of meldable heaps and randomised quicksort.
In \Cref{sec:prob-monad} the probability monad is properly introduced.
The case studies are given in \Cref{sec:case-studies} after which the soundness of our expected cost analysis is proved in~\Cref{sec:soundness}.
The remaining part of the paper showcases the use of the probability monad on the
remaining case studies. \Cref{sec:meldable-heaps} presents the expected cost analysis of meldable heaps.
Finally, after discussing related work in \Cref{sec:related-work}, we 
conclude in \Cref{Conclusion}.
Due to space restrictions, additional informations and some of the case studies have been delegated to
the Appendix.

\newsavebox\lhBegin
\newsavebox\lhEnd
\begin{lrbox}{\lhBegin}
\lstinline[style=haskell,basicstyle=\color{black}\footnotesize\ttfamily]!{-@!
\end{lrbox}

\begin{lrbox}{\lhEnd}
\lstinline[style=haskell,basicstyle=\color{black}\footnotesize\ttfamily]!@-}!
\end{lrbox}

\begin{figure}[t]
\centering
\begin{lstlisting}[style=haskell,numbers=left]
{-@ ($\triangleRight$) :: #a -> x:#b -> {r:#b | r = x} @-} ~\label{lst:triangle}~
_ $\triangleRight$ x = x

{-@
data Heap #a = Empty ~\label{lst:Empty}~
            | Heap ~\label{lst:Heap}~{ key :: #a
                   , left :: Heap {v:#a | key <= v}
                   , right :: Heap {v:#a | key <= v}
                   }

meld :: Ord #a => h1:~\Heap/~ #a -> h2:~\Heap/~ #a ->
            {result~\zsavepos{pos:meldA}~:ProperDist ({h:~\Heap/~ #a | bag h = Bag_union (bag h1) (bag h2)}) ~\label{lst:meld:bag}~
~\zsavepos{pos:meldB}\hspace{\inteval{\zposx{pos:meldA}-\zposx{pos:meldB}}sp}~| ~\expectCost/~ result <= ~\logTwo/~ (size h1) + ~\logTwo/~ (size h2)} ~\label{lst:meld:end-type}~
@-}
meld (~\Heap/~ k1 l1 r1) (~\Heap/~ k2 l2 r2)
    | k1 <= k2 =
~\hspace{1.2cm}\phantom{\triRight/}~ ~\logTwoIneq/~ (size l1) (size r1) ~\label{lst:meld:logIneq}~
~\hspace{1.2cm}\triRight/~ ~\coin/~ 0.5 
              (~\fMap/~ (~\tick/~ 1 (meld l1 (~\Heap/~ k2 l2 r2))) (\ h -> ~\Heap/~ k1 h r1))
              (~\fMap/~ (~\tick/~ 1 (meld r1 (~\Heap/~ k2 l2 r2))) (\ h -> ~\Heap/~ k1 l1 h))
\end{lstlisting}
\caption{Meldable Heaps (without base cases and the symmetric inductive case)}
\label{lst:meldableheaps}
\end{figure}

\section{Overview}
\label{sec:overview}

We begin with meldable heaps, illustrating how the introduced probability monad enables a high degree of automation for expected-cost verification.
We then discuss the central design of the cost-aware probability monad and explain how its refinement-typed operators support compositional quantitative reasoning.
Finally, we consider randomised quicksort, demonstrating how the same framework smoothly scales to mathematically more sophisticated analyses requiring interactive theorem proving.

\subsection{Meldable Heaps}

Meldable heaps~\cite{MartinezR98} are a probabilistic heap data structure supporting a \emph{meld} operation that combines two heaps into a single heap without maintaining explicit balancing invariants.
Instead, balancing is achieved probabilistically: recursive descent during melding is determined by a fair coin toss, yielding logarithmic expected cost.

\Cref{lst:meldableheaps} shows the core of our implementation together with its refinement type specification.
Our encoding starts off with a
standard type definition of heaps in a refinement type setting.
Thus, the heap constraint can be statically enforced via the annotations in line~\ref{lst:Heap}, requiring that in a heap of the form \C!Heap k l r! all values in the left and right child nodes cannot be greater than \C!k!.
The function \meld/ takes two heaps and returns a probability distribution over heaps represented by the type \ProperDist/.
The refinement type simultaneously captures both functional correctness and the expected-cost bound of the probabilistic computation.
Functional correctness is ensured at two levels: the refinement-typed heap definition already enforces the heap-order property structurally, while the refinement on the result distribution specifies that the resulting heap contains exactly the elements of the two input heaps.
In addition, the refinement type states that the expected cost of the computation is bounded by the sum of the logarithms of the heap sizes.

The implementation itself remains close to a standard executable Haskell implementation. Probabilistic choice is represented through the combinator \C!$\coin/$ 0.5!, corresponding to a fair coin toss.
Recursive calls are instrumented with the combinator \tick/, which accounts for the cost of a recursive step.
Consequently, expected costs are tracked directly inside the probabilistic computation itself rather than being handled in a separate semantic layer.

A central aspect of this example is that the induced quantitative verification conditions are almost entirely discharged automatically by \LH/.
The recursive calls to \meld/ provide the quantitative induction hypotheses for the subcomputations, while the refinement types of the monadic combinators encode how expected costs propagate through probabilistic choice and sequential composition.
As a consequence, the expected-cost recurrence of the algorithm is reflected directly in the structure of the probabilistic program.

The only required user guidance consists of explicitly supplying suitable logarithmic inequalities, such as \logTwoIneq/.
This is necessary because logarithms are treated axiomatically in our development rather than being defined directly in the SMT-supported refinement logic.
The operator \C!$\triRight/$! is used to inject such auxiliary facts into the verification of the subsequent expression.%
\footnote{\LH/ provides the {\hsBasicStyle ?} combinator, where {\hsBasicStyle (f e) ? lemma} allows to use the post-condition of {\hsBasicStyle lemma} to verify the expression {\hsBasicStyle f e}. However, if subexpression {\hsBasicStyle e} also needs the post-condition of {\hsBasicStyle lemma}, one would need to write {\hsBasicStyle (f (e ? lemma)) ? lemma}. On the other hand, in the expression {\hsBasicStyle lemma $\triangleRight$ (f e)} the post-condition of {\hsBasicStyle lemma} can be used to successfully typecheck subexpression {\hsBasicStyle e} and only one invocation of {\hsBasicStyle lemma} is necessary.}
Operationally, the operator simply evaluates to its second argument. However, its refinement type ensures that the logical facts established by the first argument become available during verification of the second argument. In the example above, the logarithmic inequality is therefore made available when verifying the expected-cost bound of the probabilistic computation.

This example already illustrates the main philosophy underlying our approach: by integrating expected costs directly into probabilistic computations through refinement-typed monadic operators, many expected-cost analyses become compositional and largely syntax-directed.
We now briefly discuss the underlying probability monad enabling this style of compositional reasoning.

\subsection{The Cost-Aware Probability Monad}

\begin{figure}[t]
\centering
\begin{lstlisting}[style=haskell]
{-@ type $\mathbb{R}^{\geqslant 0}$ = {v:#R | v >= 0} @-} ~\label{lst:rnneg}~
{-@ type Prob = {v:#R | 0 < v && v <= 1} @-}
{-@ data Outcome #a = Outcome { cost :: $\Rnneg$, value :: #a, prob :: Prob } @-} ~\label{lst:Outcome}~

data Dist #a = DEmpty | DSome (Outcome #a) (Dist #a) ~\label{lst:data-dist}~

{-@ type SubDist #a = {d:Dist #a | ~\probMass/~ d <= 1} @-} ~\label{lst:SubDist}~
{-@ type ProperDist #a = {d:Dist #a | ~\probMass/~ d = 1} @-} ~\label{lst:ProperDist}~

{-@ measure ~\DEF{probMass}~ :: ~\SubDist/~ #a -> Prob @-}
probMass ~\DEmpty/~ = 0
probMass (~\DSome/~ (~\Outcome/~ _ _ p) xs) = p + probMass xs

{-@ measure ~\DEF{expectCost}~ :: ~\SubDist/~ #a -> $\Rnneg$ @-}
expectCost ~\DEmpty/~ = 0
expectCost (~\DSome/~ (~\Outcome/~ c _ p) xs) = c * p + expectCost xs

{-@ tick :: c:$\Rnneg$ ->  d:~\SubDist/~ #a -> {r:~\SubDist/~ #a | $\HEcostTT{r}$ = c * (~\probMass/~ d) + $\HEcostTT{d}$} @-}

{-@ coin :: {p:#R | 0 <= p && p <= 1} -> success:~\ProperDist/~ #a -> failure:~\ProperDist/~ #a ->
      {r:~\ProperDist/~ #a | $\HEcostTT{r}$ = p * $\HEcostTT{success}$ + (1-p) * $\HEcostTT{failure}$} @-}
\end{lstlisting}
\caption{Essentials of the Cost-Aware Probability Monad \\ For ease of notation we write \C!$\HEcostTT{e}$! instead of \C!expectCost e! inside refinements.}
\label{fig:essentials}
\end{figure}

Following the standard monadic view of probabilistic programming, computations are represented as finite probability distributions over outcomes.
The core definitions of the monad are summarised in \Cref{fig:essentials}.
An element \C!Outcome c v p! represents a computation returning the value \C!v! with probability \C!p! while incurring cost \C!c!.
Costs are therefore tracked on individual probabilistic outcomes, allowing sequential probabilistic computations to accumulate path-dependent costs.

Using refinement types, we distinguish between subdistributions and proper distributions. A subdistribution may have total probability mass smaller than one, whereas a proper distribution has total mass exactly one. Subdistributions arise naturally as intermediate objects during probabilistic computations, for example when scaling distributions through probabilistic choice. The refinement types
\SubDist/ and \ProperDist/ statically enforce the corresponding probability-mass invariants.

The function \FN{expectCost} computes expected costs by summing the costs of all outcomes weighted by their probabilities.
Since \FN{expectCost} is implemented as a terminating Haskell function, Liquid Haskell's refinement reflection mechanism can lift it into the refinement logic, allowing expected costs to appear directly inside refinement specifications and proofs.\footnote{In \LH/, this is achieved by annotationg a function with \texttt{\textbf{measure}}. Such functions are called \emph{measures} and they can be used to effectively specify properties of data structures since the values of measures are automatically tracked by \LH/ (see \cite{conf/icfp/VazouSJVJ14, tp-for-all} for more details).}

The crucial aspect of the monad is that the refinement types of its combinators already encode quantitative reasoning principles.
For example, the refinement type of \tick/ specifies that incrementing all outcomes by a fixed cost \C!c! increases the expected cost of a proper distribution by exactly \C!c!; more generally, for subdistributions, the increment is weighted by the probability mass of the distribution.

Similarly, the refinement type of \coin/ captures that the expected cost of a probabilistic choice is given by the weighted average of the expected costs of the two branches.

Consequently, quantitative reasoning follows the structure of the probabilistic program itself.
In the meldable-heaps example, the recursive calls provide the induction hypotheses for the subcomputations, while the refinement types of \tick/ and \coin/ encode how the expected costs compose.
\LH/ can therefore derive the induced expected-cost recurrence compositionally from the refinement-typed monadic structure of probabilistic programs rather than from a separate semantic cost model.

The use of finite discrete distributions plays an important role in enabling this style of verification.
Since expectations ultimately reduce to finite summations over explicit outcomes, the resulting verification conditions interact naturally with SMT-supported refinement typing while preserving executability of the probabilistic programs.

In the next subsection, we illustrate how the same framework also supports mathematically more sophisticated analyses through the mechanisation of randomised quicksort.

\subsection{Randomised Quicksort}
\label{sec:overview:quicksort}

\begin{figure}[t]
\begin{lstlisting}[style=haskell,numbers=left]
data ~\DEF{List}~ #a = Nil | Cons #a (List #a)

{-@ rquick :: Ord #a => {l:~\List/~ #a | ~\noDuplicates/~ l} -> ~\label{fig:quick-sort-simplified:rquick}~
    ~\ProperDist/~ {r:~\List/~ #a | sorted r && elems l = elems r} @-} ~\label{fig:quick-sort-simplified:rquick-sig}~
rquick ~\Nil/~ = ~\dunit/~ ~\Nil/~
rquick l = ~\uniform/~ 0 (length l - 1) (rquickBody l)

{-@ rquickBody :: Ord #a => {l:~\List/~ #a | ~\noDuplicates/~ l} -> {k:#N | k < length l} ->
        ~\ProperDist/~ {r:~\List/~ #a | sorted r && elems l = elems r} @-} ~\label{fig:quick-sort-simplified:rquickBody-sig}~
rquickBody l k = let e = ref l k in
                 let (smaller, bigger) = ~\partition/~ e (deleteAt l k) in ~\label{fig:quick-sort-simplified:rquickBody-def}~
                 ~\tick/~ (length l - 1) (~\combine/~ (rquick smaller) (rquick bigger) (merge e))
\end{lstlisting}
\caption{Randomised Quicksort (Simplified).}
\label{fig:quick-sort-simplified}
\end{figure}

We now consider randomised quicksort, illustrating how the proposed framework also supports mathematically more sophisticated expected-cost analyses requiring interactive theorem proving.
\Cref{fig:quick-sort-simplified} shows the core of our implementation.\footnote{For ease of presentation the depicted code has been simplified and shortened. In particular, termination measures have been elided.
The complete code is depicted in \Cref{fig:quick-sort} in \Cref{sec:rand-quick}.}

As in the \isabelle{} formalisation~\cite{EberlHN20}, we assume that the input list contains no duplicate elements, simplifying the probabilistic analysis through unique ranks of pivot elements.
Our formalisation establishes that the expected number of comparisons performed by randomised quicksort on a list with $n$ distinct elements is equal to $2(n+1)H_n - 4n$, where $H_n$ denotes the $n$-th harmonic number, confirming the results of the \isabelle{} proof.
The function \rquick/ samples a pivot index uniformly at random, partitions the remaining list into smaller and larger elements, recursively sorts the two partitions, and finally combines the results.

As in the meldable-heaps example, the implementation remains close to a standard executable Haskell program.
Following~\cite{EberlHN20}, we account for the \C!length l - 1! comparisons performed during partitioning through a single outer application of \tick/, rather than directly instrumenting the partitioning procedure itself.
While the latter would also be possible within our framework, the chosen formulation simplifies the subsequent formalisation and facilitates comparison with the \isabelle{} development.

The refinement type of \rquick/ expresses functional correctness by specifying that the resulting lists are sorted and contain exactly the elements of the input list, and is automatically verified by \LH/.
On the other hand, inferring a closed-form solution for the expected number of comparisons requires manual reasoning.
This is hardly surprising as the analysis requires substantial arguments about finite summations, permutations, and harmonic numbers.

The formalisation establishing the closed form solution for the expected number of comparisons performed by \rquick/ is divided into two steps.
%

In the first step, we establish that the expected cost of \rquick/ satisfies the classical \emph{quicksort recurrence} \C!qs_rec!, which is realised as the Haskell function
\begin{lstlisting}[style=haskell]
qs_rec 0 = 0
qs_rec n = n - 1 + 1/n * $\HfinSumTT{0}{n-1}$ (\k -> qs_rec k + qs_rec (n-1-k))
\end{lstlisting}
where we make use of the finite-sum operator \C!$\finSumSymbOnly$! that we introduce as part of our finite-sum library.

We remark on the level of automation offered by our cost-aware probability monad.
The local probabilistic expected-cost bookkeeping is handled compositionally through the monadic structure of the implementation itself:
The combinator \uniform/ contributes the averaging over all pivot choices, \tick/ contributes the additive partitioning cost \C!length l - 1!, and \combine/ propagates the recursive expected costs of the two recursive calls.
Consequently, the framework automatically derives the expected-cost equation
\begin{lstlisting}[style=haskell]
$\HEcostOpen$rquick l$]$ = n-1 + $\frac{\mathtt{1}}{\mathtt{n}}$ $\HfinSumTT{0}{n-1}$ (k -> $\HEcostOpen$sm l k$]$ + $\HEcostOpen$bg l k$]$),
\end{lstlisting}
where \C!(sm l k, bg l k) = partition (ref l k) (deleteAt l k)! denotes the result of partitioning \C!l! with regard to index \C!k!, \C!n = length l! denotes the length of the input list, and the summation ranges over the concrete pivot indices \C!k! of the list.

However, the obtained expression still depends explicitly on the concrete pivot positions. To derive the \C!qs_rec! recurrence, this dependence must be eliminated and replaced by a dependence only on the sizes of the recursive subproblems.
Following~\cite{EberlHN20}, we therefore introduce ranks of pivot elements and establish an index-rank bijection showing that the recursive costs only depend on the rank of the chosen pivot element rather than its concrete position in the list.
Using our finite-sum library, we subsequently mechanise the required symmetry arguments and summation rewritings to obtain the pivot-independent recurrence \C!qs_rec! stated above.

The result of this step is the theorem with the following refinement type:
\begin{lstlisting}[style=haskell]
{-@ thm_rquick_cost_eq_recurrence :: Ord #a => {l:List #a | ~\noDuplicates/~ l} ->
        { $\HEcostOpen$rquick l$]$ = qs_rec (length l) } @-}
\end{lstlisting}

The second step consists of the derivation of a closed form solution of \C!qs_rec!, as witnessed by the theorem \C!thm_qs_rec_closed_form! below, where \C!$\harmonic/$ n! computes the \C!n!-th harmonic number:
%
\begin{lstlisting}[style=haskell]
{-@ thm_qs_rec_closed_form :: n:#N -> { qs_rec n = 2*(n+1) * ~\harmonic/~ n - 4*n } @-}
\end{lstlisting}
The final theorem establishing the closed form for the number of expected comparisons of randomised quicksort then just invokes the aforementioned steps as shown below.
\begin{lstlisting}[style=haskell]
{-@ thm_rquick_cost :: Ord a => {l:List a | ~\noDuplicates/~ l} ->
    { $\HEcostOpen$rquick l$]$ = 2*(length l + 1) * ~\harmonic/~ (length l) - 4*(length l) } @-}
thm_rquick_cost l = ~\zsavepos{pos:rquickCostA}\phantom{\triRight/}~ thm_rquick_cost_eq_recurrence l -- Step 1
~\zsavepos{pos:rquickCostB}\hspace{\inteval{\zposx{pos:rquickCostA}-\zposx{pos:rquickCostB}}sp}\triRight/~ thm_qs_rec_closed_form (length l) -- Step 2
\end{lstlisting}

\section{Cost-Aware Probability Monad}
\label{sec:prob-monad}

Probabilistic algorithms, i.e., functions whose output is a probability distribution, are usually implemented using \emph{probability monads} in Haskell (cf. \cite{Ramsey02, Erwig06, Scibior15, VasilenkoVB22}).
This approach enables the succinct representation of randomised programs by structuring computations using a monad.
We implement our cost-aware probability monad for finite discrete probability distributions and use \LH/ to provide support for auto-active expected cost and expected value analysis of probabilistic algorithms.
In order to guarantee soundness of the facts established by refinement types, we require \LH/ to prove termination of all defined functions.

In the following, we give an overview over the core functionality of the cost-aware probability monad.
We end this section by introducing a library enabling the reasoning about finite sums within \LH/ and illustrate how we effectively deal with logarithms.

\paragraph{Functions of the Cost-Aware Probability Monad}

The core datatypes of the developed cost-aware probability monad were already
presented in~\Cref{fig:essentials}. In the following we introduce the functions---and
sketch their implementation in the monad---%
used as building blocks for the implementation and verification of probabilistic algorithms.
These functions can be used in refinements using the \emph{refinement reflection} feature of \LH/ \cite{ref-refl}.
Refinement reflection denotes the process of lifting a \emph{terminating} Haskell function to the refinement logic whenever it is annotated with a \C!reflect! annotation, such that it can be used in refinement types.\footnote{Refinement reflection is more general than the measure functionality of \LH/ but values of reflected functions are not automatically tracked.}
This enables the use of \LH/ as a theorem prover as reflected functions can be reasoned about in refinement types.
A theorem is then a Haskell function that takes premises as input and returns the unit type \C!()! that is refined with a logical predicate representing the theorem's statement.
For instance, the following theorem shows that two \tick/ statements can be merged into a single one:
\begin{lstlisting}[style=haskell]
{-@ thm_tick_merge :: c1:$\Rnneg$ -> c2:$\Rnneg$ -> d:~\SubDist/~ #a ->
        { ~\tick/~ c1 (~\tick/~ c2 d) = ~\tick/~ (c1+c2) d } @-}
thm_tick_merge _ _ ~\DEmpty/~ = ()
thm_tick_merge c1 c2 (~\DSome/~ _ xs) = thm_tick_merge c1 c2 xs
\end{lstlisting}
Note that the base case for empty distributions is verified automatically and we can just provide the unit value \C!()!, whereas for non-empty distributions we can call the theorem recursively in order for \LH/ to verify the refinement type.
Moreover, termination is automatically established from the structural recursion on the shape of the distribution.
Hence, inductive proofs are encoded as terminating recursive functions in \LH/.
For this reason it is crucial to ensure that we reason about \emph{finite} probability distributions in order to guarantee the soundness of properties encoded in refinement types.
We follow the convention to prefix the names of \LH/ theorems with \C!thm! in order to distinguish them from regular Haskell code.
Since all the functions in the remainder of this section are reflected, we will omit the \C!reflect! annotations.

\begin{figure}[t]
\begin{lstlisting}[style=haskell]
{-@ ~\DEF{expectVal}~ :: (#a -> #R) -> ~\SubDist/~ #a -> #R @-}
expectVal e ~\DEmpty/~  = 0
expectVal e (~\DSome/~ (~\Outcome/~ c x p) xs) = (e x) * p + expectVal e xs

{-@ ~\DEF{probOf}~ :: Eq #a => #a -> d:~\SubDist/~ #a -> {p:$\Rnneg$ | p <= ~\probMass/~ d} @-}
probOf v ~\DEmpty/~ = 0
probOf v (~\DSome/~ (~\Outcome/~ c x p) xs) = (if x == v then p else 0) + (probOf v xs)
\end{lstlisting}

\caption{Expected Values and Probabilities}
\label{fig:expect-val}
\end{figure}

Similar to expected costs, we define the \emph{expected value} of a distribution, as well as the probability of a value as defined in~\Cref{fig:expect-val}.
The first argument of \expectVal/ is a function mapping a value of type \C!$\alpha$! to a real number, a so-called \emph{expectation}.
We write \C!$\HEvalTT{e}{d}$! for \C!expectVal e d!.
The probability of a value in a given distribution is formalised by the function \probOf/.


In many cases we are able to precisely state using refinement types how the expected cost, expected value and probability mass of the resulting distribution behaves in relation to the input parameters.
These annotations are crucial in providing a high degree of automation, when the core functions of the monad are used.
As an expectation is needed to compute the expected value of a distribution, we use so-called \emph{ghost parameters}, cf.~\cite{Fillitre2016}.
Ghost parameters---also called auxilliary variables in the literature, cf.~\cite{Kleymann99,AvanziniBGMV24}---%
are commonly used in program verification. They are never evaluated by the function and their only purpose is the facilitation of stronger properties in the refinement types.
In contexts where expected values are not relevant, we omit the ghost parameter for the sake of readability.

As all the functions are implemented by simple recursion on the sequence of outcomes comprising the input distributions, we only state the more interesting refinement types which are given in~\Cref{fig:probMonad-defs}.
Note that several of the functions take as a ghost parameter an expectation \C!e! as first input in order to specify the behaviour of the expected value.
Kindly note that the \C!tick! function in~\Cref{fig:essentials} is actually the \C!tick! function from~\Cref{fig:probMonad-defs}, where the expectation argument is elided for readability.

\begin{figure}
\begin{lstlisting}[style=haskell,xleftmargin=0pt]
~\DEF{dunit}~ :: e:(#a -> #R) -> x:#a -> {d:~\ProperDist/~ #a | $\HEcostTT{d}$ = 0 && $\HEvalTT{e}{d}$ = e x && ~\probOf/~ x d = 1}

tick :: e:(#a -> #R) -> c:#Rnneg -> d:~\SubDist/~ #a -> ~\label{lst:tick}~
    {r:~\SubDist/~ #a | ~\probMass/~ r = ~\probMass/~ d && $\HEcostTT{r}$ = c * (~\probMass/~ d) + $\HEcostTT{d}$ && $\HEvalTT{e}{d}$ = $\HEvalTT{e}{r}$}

pMap :: e:(#a -> #R) -> {p:#R | 0 < p && p <= 1} -> d:~\SubDist/~ #a -> ~\label{lst:pMap}~
    {r:~\SubDist/~ #a | ~\probMass/~ r = p * (~\probMass/~ d) && $\HEcostTT{r}$ = p * $\HEcostTT{d}$ && $\HEvalTT{e}{r}$ = p * $\HEvalTT{e}{d}$}

fMap :: e:(#b -> #R) -> f:(#a -> #b) -> d:~\SubDist/~ #a -> ~\label{lst:fMap}~
    {r:~\SubDist/~ #b | ~\probMass/~ r = ~\probMass/~ d && $\HEcostTT{r}$ = $\HEcostTT{d}$ && $\HEvalTT{e}{r}$ = $\HEval{\mathtt{e} \circ \mathtt{f}}{\mathtt{d}}$}

append :: e:(#a -> #R) -> d1:~\SubDist/~ #a -> {d2:~\SubDist/~ #a | (~\probMass/~ d1) + (~\probMass/~ d2) <= 1} -> ~\label{lst:append}~
    {r:~\SubDist/~ #a | ~\zsavepos{pos:appendA}\probMass/~ r = (~\probMass/~ d1) + (~\probMass/~ d2) &&
~\zsavepos{pos:appendB}\hspace{\inteval{\zposx{pos:appendA}-\zposx{pos:appendB}}sp}~$\HEcostTT{r}$ = $\HEcostTT{d1}$ + $\HEcostTT{d2}$ && $\HEvalTT{e}{r}$ = $\HEvalTT{e}{d1}$ + $\HEvalTT{e}{d2}$}

bind :: e:(#a -> #R) -> d:~\SubDist/~ #a -> (x:#a -> {d:~\ProperDist/~ #b | $\HEcostTT{d}$ = e x}) -> ~\label{lst:bind}~
    {r:~\SubDist/~ #b | ~\probMass/~ r = ~\probMass/~ d && $\HEcostTT{r}$ = $\HEcostTT{d}$ + $\HEvalTT{e}{d}$}

combine :: f:(#a -> #b -> #c) -> d1:~\SubDist/~ #a -> d2:~\SubDist/~ #b -> ~\label{lst:combine}~
    {r:~\SubDist/~ #c | ~\zsavepos{pos:combineA}\probMass/~ r = (~\probMass/~ d1) * (~\probMass/~ d2) &&
~\zsavepos{pos:combineB}\hspace{\inteval{\zposx{pos:combineA}-\zposx{pos:combineB}}sp}~$\HEcostTT{r}$ = (~\probMass/~ d2) * $\HEcostTT{d1}$ + (~\probMass/~ d1) * $\HEcostTT{d2}$}
\end{lstlisting}

\caption{Functions of the Probability Monad}
\label{fig:probMonad-defs}
\end{figure}

The function \dunit/ returns a cost-free Dirac distribution whose only outcome contains the given value.
Costs are incurred using \tick/, which increments the cost of each outcome of the given distribution by a non-negative real value.
Its refinement type precisely captures how probability mass, expected cost and expected value of the resulting distribution behaves in relation to the inputs.
Similarly, the function \pMap/ scales the probability of each outcome by the given value $p \in (0,1]$.
As a consequence, the probability mass, expected cost and expected value of \C!pMap p d! are scaled by \C!p!.
While \tick/ and \pMap/ modify costs and probabilities, respectively, \fMap/ is used to transform the values inside a distribution by a given function.
The expected cost and probability mass of \C!fMap g f d! stay the same, whereas the expected value of the resulting distribution can be characterised by the composition of the provided expectation \C!g! with the function \C!f!.
Simple combination of distributions is realised by just appending the sequence of outcomes of one distribution after another sequence of outcomes, as implemented by \append/.
Probability mass, expected cost and expected value of \C!append e d1 d2! are given by summing up the respective values of \C!d1! and \C!d2!.

Sequential composition is realised through the monadic \bind/ function, which also takes care of threading incurred costs through the computation.
The first argument of \bind/ is a ghost parameter expectation that is used to specify the expected cost of the resulting distribution.
That is, in the expression \C!$\bind/$ e d f!, the refinement type annotation requires that the expected cost of \C!f x! for each input value \C!x! is equal to \C!e x!.
The post-condition then states that the expected cost of \C!bind e d f! is equal to the expected cost of \C!d! in addition to the \emph{expected value} of \C!d! with respect to expectation \C!e!.

Ultimately, the function \C!combine f d1 d2! can be thought of a cartesian product of distributions \C!d1! and \C!d2!.
That is, for every outcome \C!o1! of \C!d1! and \C!o2! of \C!d2! the function \combine/ adds the following outcome to the resulting distribution:
\begin{lstlisting}[style=haskell]
Outcome{cost=(cost o1)+(cost o2), value=f (value o1) (value o2), prob=(prob o1)*(prob o2)}
\end{lstlisting}
The refinement type for the distribution computed by \combine/ holds since costs of outcomes are added and probabilities are multiplied.


We model probabilistic choices using the \C!coin! function illustrated next, which again takes an expectation as input to track the expected value of the resulting distribution.
%
\begin{lstlisting}[style=haskell]
{-@ ~\DEF{coin}~ :: e:(#a -> #R) -> {p:#R | 0 <= p && p <= 1} ->
        success:~\ProperDist/~ #a -> failure:~\ProperDist/~ #a ->
        {r:~\ProperDist/~ #a | ~\zsavepos{pos:coinA}~$\HEcostTT{r}$ = p * $\HEcostTT{success}$ + (1-p) * $\HEcostTT{failure}$ &&
~\zsavepos{pos:coinB}\hspace{\inteval{\zposx{pos:coinA}-\zposx{pos:coinB}}sp}~$\HEvalTT{e}{r}$ = p * $\HEvalTT{e}{success}$ + (1-p) * $\HEvalTT{e}{failure}$} @-}
coin e 0 success failure = failure
coin e 1 success failure = success
coin e p success failure = ~\append/~ e (~\pMap/~ e p success) (~\pMap/~ e (1-p) failure)
\end{lstlisting}
The intuitive meaning of \C!coin e p success failure! is that with probability \C!p! the expression \C!success! is chosen and \C!failure! with probability \C!1-p!.

To support reasoning about uniformly random integers, we define the \uniform/ function, which builds upon \coin/.
By taking a continuation \C!f :: #Z -> ProperDist #a! as input, we can specify the expected cost of \C!uniform lo hi f! by the sum of the expected costs of \C!f i! for \C!lo <= i <= hi! divided by \C!hi-lo+1!, i.e., the weighted average of the resulting distributions.
A corresponding property also holds true for the expected value with respect to a given expectation.
With the help of the operator \C!$\finSumSymbOnly$! of the finite-sum library we can represent these facts by stating an expressive refinement type annotation for \C!uniform!, where \C!$\HfinSumTT{lo}{hi}\,$f! evaluates to \C!f lo + f (lo+1) + $\ldots$ + f hi!.%

\noindent\begin{minipage}{\linewidth}
\begin{lstlisting}[style=haskell]
{-@ type ~\DEF{Range}~ LO HI = {num:#Z | LO <= num && num <= HI} @-}

{-@ ~\DEF{uniform}~ :: e:(#a -> #R) -> lo:#Z -> {hi:#Z | lo <= hi} -> f:(~\Range/~ lo hi -> ~\ProperDist/~ #a) ->
        {r:~\ProperDist/~ #a | ~\zsavepos{pos:uniformA}~$\HEcost{\mathtt{r}}$ = (1 / (hi-lo+1)) * $\HfinSumTT{lo}{hi}$ ($\HEcostOnly$ $\circ$ f) &&
~\zsavepos{pos:uniformB}\hspace{\inteval{\zposx{pos:uniformA}-\zposx{pos:uniformB}}sp}~$\HEvalTT{e}{r}$ = (1 / (hi-lo+1)) * $\HfinSumTT{lo}{hi}$ ($\HEvalTTOnly{e}$ $\circ$ f)} / [hi-lo] @-}
uniform e lo hi f |  lo == hi = f lo
                  | otherwise = ~\coin/~ e (1 / (hi-lo+1)) (f lo) (uniform e (lo+1) hi f)
\end{lstlisting}
\end{minipage}

As in previous definitions, the expectation passed to \C!uniform! is a ghost parameter as it is only present in order to establish the property specifying the behaviour of the expected value.
In contexts where we are only interested in expected costs, we elide this parameter for readability.
Note that it is automatically established by \LH/ that \C!uniform! always results in a proper distribution and also that the expected cost and expected value is the average of the sum of distributions computed by the continuation using the uniform choices.
It is crucial to restrict the possible inputs for the continuation to \C!Range lo hi!, as \uniform/ is often used to generate uniformly random list indices, as in our randomised quicksort case study, for instance.
The refinement types must then keep track of the correct bounds in order to use the indices for list operations.

\paragraph{Monad Laws}
There is a convention in Haskell that implementation of typeclass should adhere to certain algebraic laws.
In the case of monads, our implementation has to satisfy the following rules, which we proved within \LH/:
\begin{itemize}[itemsep=2pt,topsep=2pt]
    \item Left identity: \C!$\bind/$ ($\dunit/$ x) f = f x!.
    \item Right identity: \C!$\bind/$ d $\dunit/$ = d!.
    \item Associativity: \C!$\bind/$ ($\bind/$ m g) h = $\bind/$ m (\x -> $\bind/$ (g x) h)!
\end{itemize}
In addition to these laws, the probability monad exhibits several more laws that we can exploit for the analysis of programs.
For instance, consider the program \C!$\bind/$ ($\coin/$ e p d1 d2) g!, for which expected costs need to be analysed.
Using a law of the probability monad, it can be rewritten into the equivalent form \C!$\coin/$ e p ($\bind/$ d1 g) ($\bind/$ d2 g)!.

Additionally, it holds that \C!fMap e d f! is equal to \C!bind d (\x -> dunit (f x))!.
However, by using \C!fMap! we get a stronger refinement type for the resulting distribution as the expected costs and expected values are tracked.
Similarly, \C!combine! can also be represented by \C!bind! and \C!dunit!, but we would again lose the properties stated in the refinement type.
This is the reason that we do not employ \C!bind! in the case studies and rather use the specialised functions facilitating a higher degree of automation.

\subsection{Support for Summations and Logarithms}
\label{sec:sums-logs}

In addition to the presented cost-aware probability monad we developed self-contained modules that enable reasoning about finite summations as well as logarithms within \LH/.
While the definitions for finite summations can be reflected to the logic level we make use of an axiomatic treatment of logarithmic functions.

\begin{figure}
\begin{lstlisting}[style=haskell]
{-@ reflect finSum :: lo:#Z -> hi:#Z -> f:(Range lo hi -> #R) -> ~\label{lst:finSum}~
        {r:#R | (lo <= hi) => (r = f lo + finSum (lo + 1) hi f && ~\label{lst:finSum:ref1}~
                               r = finSum lo (hi-1) f + f hi)} ~\label{lst:finSum:ref2}~
        / [if lo <= hi then hi-lo else 0] @-}
finSum lo hi f | lo == hi = f lo
               |  lo < hi = finSum lo (hi-1) f |> f lo + finSum (lo+1) hi f
               | otherwise = 0
\end{lstlisting}
\caption{Implementation of Finite Summations through \C!finSum!}
\label{fig:finSum}
\end{figure}

\subsubsection{Summations}
\label{sec:summations}

By defining a function \C!finSum! as depicted in \Cref{fig:finSum} and reflecting it to the refinement level, we enable the use of this function in refinement types.
For readability we abbreviate \C!finSum lo hi f! with \C!$\HfinSumTT{lo}{hi}\,$f!.
Note that in the implementation of \C!finSum! we used the triangle operator \C!$\triRight/$! to establish the refinement types specifying the fact that the summation can be unrolled from both sides, cf. lines~\ref{lst:finSum:ref1} and~\ref{lst:finSum:ref2} of \Cref{fig:finSum}.
In addition, we provide a small collection of helpful theorems that can be used to manipulate and rewrite finite sums.
For example, consider the following two theorems:
\vspace{-0.25cm}
\begin{center}
\begin{minipage}[t]{0.43\linewidth}%
\begin{lstlisting}[style=haskell,xleftmargin=0pt,aboveskip=0pt,belowskip=0pt]

~\DEF{thm_sum_split}~ :: lo:#Z ->
  {hi:#Z | lo <= hi} -> k:~\Range/~ lo hi ->
  f:(~\Range/~ lo hi -> #R) -> 
  { $\HfinSumTT{lo}{hi}$ f = $\HfinSumTT{lo}{k-1}$ f + f k + $\HfinSumTT{k+1}{hi}$ f }
\end{lstlisting}
\end{minipage}%
\hfill\vrule width 0.6pt\hfill
\begin{minipage}[t]{0.56\linewidth}%
\begin{lstlisting}[style=haskell,xleftmargin=7pt,aboveskip=0pt,belowskip=0pt]
~\DEF{thm_sum_permute}~ :: lo:#Z -> hi:#Z ->
  f:(~\Range/~ lo hi -> #R) ->
  $\pi$:(~\Range/~ lo hi -> ~\Range/~ lo hi) ->
  $\tau$:(~\Range/~ lo hi -> ~\Range/~ lo hi) ->
  (x:~\Range/~ lo hi -> {$\pi$ ($\tau$ x) = x && $\tau$ ($\pi$ x) = x}) ->
  { $\HfinSumTT{lo}{hi}\,$f = $\HfinSumTT{lo}{hi}\,$(f $\circ$ $\pi$) }
\end{lstlisting}
\end{minipage}
\end{center}
\vspace{0.5em}

The theorem \C!thm_sum_split! on the left-hand-side allows to split a sum at a given index while the theorem \C!thm_sum_permute! establishes that the sum indices can be arbitrarily permutated.
This is achieved by the parameters \C!$\pi$! and \C!$\tau$! togehter with a function that shows that for any given \C!x! the functions \C!$\pi$! and \C!$\tau$! are inverse functions of one another.
A complete list of theorems is given in the Appendix.

\subsubsection{Logarithms}
\label{sec:logarithms}
Since \LH/ uses the SMT theory of real numbers to reason about Haskell's \C!Double! floating point values, we cannot implement the real mathematical logarithm as a Haskell function that exhibits the same theoretical properties.

We support verification of logarithmic properties by introducing the \emph{uninterpreted} function \logTwo/ on the logic level together with a minimal set of 3 axioms satisfied by \logTwo/, as given in~\Cref{fig:log-axioms}.
An uninterpreted function can be declared on the logic level using the \C!measure! keyword and axioms can be stated through functions that are annotated with the \C!assume! keyword.
Note that in these assumptions it is ensured that \logTwo/ is only applied to positive arguments.

\begin{figure}[t]
\begin{lstlisting}[style=haskell]
{-@ type $\mathbb{Z}^+$ = {i:Int | i > 0} @-} ~\label{lst:zpos}~
{-@ measure log~\textsubscript{2}~ :: $\Zpos$ -> $\Rnneg$ @-} ~\label{lst:logTwo}~

{-@ assume     log~\textsubscript{2}~_base~\label{lst:log2_base}~ :: {~\logTwo/~ 2 = 1} @-}
{-@ assume  log~\textsubscript{2}~_product~\label{lst:log2_product}~ :: x:$\Zpos$ -> y:$\Zpos$ -> {~\logTwo/~ (x*y) = ~\logTwo/~ x + ~\logTwo/~ y} @-}
{-@ assume log~\textsubscript{2}~_monotone~\label{lst:log2_monotone}~ :: x:$\Zpos$ -> {y:$\Zpos$ | x < y} -> {~\logTwo/~ x < ~\logTwo/~ y} @-}
\end{lstlisting}

\caption{Logarithm Axiomatisation in \LH/}
\label{fig:log-axioms}
\end{figure}
%

From this set of axioms we then prove properties such as the following:
\begin{lstlisting}[style=haskell]
{-@ thm_log~\textsubscript{2}~_ge_0 :: x:$\Zpos$ -> { ~\logTwo/~ x >= 0 } @-} ~\label{lst:thm_log2_ge_0}~
thm_log~\textsubscript{2}~_ge_0 x = if x == 1 then () else ~\logTwoMonotone/~ 1 x
\end{lstlisting}

Moreover the next logarithmic inequality will become essential for deriving cost bounds for meldable heaps in~\Cref{sec:meldable-heaps} as well as for randomised splay trees in~\Cref{sec:splay-trees} (cf.~\citet{Okasaki:1999}).
The full proof is explained in~\Cref{appendix:log}.
\begin{lstlisting}[style=haskell]
{-@ thm_log~\textsubscript{2}~_inequality :: x:$\Zpos$ -> y:$\Zpos$ -> {2 + ~\logTwo/~ x + ~\logTwo/~ y <= 2 * ~\logTwo/~ (x+y)} @-} ~\label{lst:thm_log2_inequality}~
\end{lstlisting}

\section{Case Studies}
\label{sec:case-studies}

We evaluate the cost-aware probability monad with respect to the ease of verification of both quantitative and qualitative properties of probabilistic programs.
To this end, we present the implementation and verification of several classical probabilistic algorithms and data structures: 
\emph{meldable heaps}, \emph{randomised quicksort}, \emph{randomised splay trees}, \emph{random permutations}, and the \emph{hiring problem}. 
These case studies demonstrate the probability monad's ability to succinctly represent probabilistic computations.
Additionally, the usage of the monad makes it possible to seamlessly establish the functionial correctness of programs without any overhead by exploiting the rich refinement type system of \LH/ together with SMT-supported automation.

In addition to the case studies presented in this section, we describe a concise formalised proof of \emph{Bayes' theorem} as well as the verification of \emph{randomised quickselect} in the Appendix.

\subsection{Meldable Heaps}
\label{sec:meldable-heaps}

We exemplify (nearly) fully automated expected cost analysis enabled by the probability monad via the implementation of a simple probabilistic data structure, namely \emph{meldable heaps} \cite{GambinM98}.

\begin{figure}
\begin{lstlisting}[style=haskell,numbers=left]
{-@ ~\DEF{meld}~ :: Ord #a => h1:~\Heap/~ #a -> h2:~\Heap/~ #a ->
            {d:~\ProperDist/~ ({h:~\Heap/~ #a | bag h = Bag_union (bag h1) (bag h2)}) ~\label{lst:meld-functional-correct}~
              | $\HEcost{\mathtt{d}}$ <= ~\logTwo/~ (size h1) + ~\logTwo/~ (size h2)} / [size h1 + size h2] @-} ~\label{lst:meld:termination}~
meld h1@~\Empty/~ h2 = ~\logTwoGeZero/~ (size h1) |> ~\logTwoGeZero/~ (size h2) |> ~\dunit/~ h2
meld h1 h2@~\Empty/~ = ~\logTwoGeZero/~ (size h1) |> ~\logTwoGeZero/~ (size h2) |> ~\dunit/~ h1
meld (~\Heap/~ k1 l1 r1) (~\Heap/~ k2 l2 r2) | k1 <= k2 = -- ~\hsCommentStyle{} see \Cref{lst:meldableheaps}~
                                     | k2 < k1 = -- ~\hsCommentStyle{} Symmetric to the previous case~
\end{lstlisting}
\caption{Randomised Meld Function}
\label{fig:meld}
\end{figure}

Meldable heaps derive their name from their central operation of \emph{melding}
that suitably combines two heaps $h_1$ and $h_2$ into a new one. The innovation of
\emph{meldable heaps} is to perform this operation subject to a fair coin toss, resulting
in an (expected) bound of the cost of the operation of $\log_2(\size{h_1}) + \log_2(\size{h_2})$, where $\size{h}$ denotes the number of leaf nodes of the heap $h$.

Our contributions in this case study are twofold.
First, we show the functional correctness of melding in a fully automated way.
Second, through providing two hints to the solver regarding properties of logarithms,
the optimal upper bound on the expected cost is established by \LH/ without further guidance.
In this sense, our probability monad provides the user of \LH/ with significantly increased
automation support for the analysis of a typical (probabilistic) data structure, cf.~\cite{Matheja20,LMZ:2022}.

\paragraph{Functional Correctness.}
The implementation of meldable heaps is given in \Cref{fig:meld}.
The function \meld/ implements the melding functionality, making use of the \coin/ function to represent fair coin tosses.

In order to show termination of the \meld/ function, we need to provide a termination metric, as recursion happens on both input heaps and \LH/ only performs automatic termination checking on a single parameter.
Using the measure functionality of \LH/ we define a measure \C!size! that computes the number of leaves of a heap.

To this end, the expression \C!/ [size h1 + size h2]! in line~\ref{lst:meld:termination} instructs \LH/ to use the sum of the heaps sizes as a \emph{termination metric}.
In general, a termination metric has the form \C!/ [$e_1$, $\ldots$, $e_n$]! and can be provided by the user when \LH/ cannot automatically establish termination.
Every expression $e_i$ can refer to the function arguments and must always evaluate to a natural number.
Furthermore, for each recursive call, the expressions $(e_1, \ldots, e_n)$ must decrease lexicographically, which proves termination \cite{tp-for-all}.

As \meld/ either recurses on \C!h1! or \C!h2! it holds that \C!size h1 + size h2! decreases in every recursive call
Note that it is crucial to state that \C!size! always results in a positive integer because \LH/ can then automatically verify that the sizes of the left and right sub-heaps are strictly smaller than the overall size of the heap.

The refinement type of the result of \meld/ states that for any ordered type of elements $\alpha$, the melding of two heaps over such
elements results in a distribution of heaps of the union the elements in the two heaps (cf.~\Cref{fig:meld}, line~\ref{lst:meld-functional-correct}).
The \C!bag! function is a measure that computes a multiset containing all elements of a given heap.
Moreover, \C!Bag_union! is provided by \LH/ to denote multiset union.
Due to the SMT-based automation, the functional correctness can be established fully automatically.

\paragraph{Expected Cost Analysis.}
With the combination of our monad operators and \LH/'s automation, we are able to derive the optimal upper bound of the expected cost of \meld/ with minimal overhead, where each of the recursive calls is accounted for with cost~$1$.
Since the functions provided by the probability monad keep track of incurred costs we only need to provide guidance regarding logarithmic reasoning and no theorems about the algorithm are required for the expected-cost analysis, cf.~\Cref{table:loc}.

If either of the two heaps is empty, the expected cost of \meld/ is 0.
Since \logTwo/ is uninterpreted, we need to give a hint to \LH/ that the logarithms of the heap sizes are non-negative.

It remains to show the validity of the bound for the scenario where both heaps are non-empty.
Let \C!h1 = $\Heap/$ k1 l1 r1!, and \C!h2 = $\Heap/$ k2 l2 r2!.
If \C!k1 <= k2!, the heap \C!h2! can be melded with either \C!l1! or \C!r1!.
Due to the refinement annotations tracking expected costs and by using the cost bound obtained from the recursive calls, \LH/ is able to deduce that
{\footnotesize \( \HEcostOpen\meld/\ \mathtt{h1\ h2}] \leqslant 1 + 0.5 \cdot \logTwo/\ (\mathtt{size\ l1}) + 0.5\cdot \logTwo/\ (\mathtt{size\ r1}) + \logTwo/\ (\mathtt{size\ h2}) \)}.
By applying \C!$\logTwoIneq/$ (size l1) (size r1)! (see line~\ref{lst:meld:logIneq} in \Cref{lst:meldableheaps}) and the fact that \C!size h1! is equal to \C!size l1 + size r1!, we arrive at the desired optimal upper bound of
{\footnotesize \( \logTwo/\ (\mathtt{size\ h1}) + \logTwo/\ (\mathtt{size\ h2}). \)}

\paragraph{Related Work}

\citet{LMZ:2022} employ meldable heaps and randomised splay trees as
case studies, where the precises expected amortised cost can be inferred
fully automatically by their prototype implementation \atlas. Our formalisation
of meldable heaps establishes the same precise bounds, that can be checked automatically,
leveraging the automation support of \LH/.
\citet{HaselwarterLMG024} use \tachis{} to verify the functional correctness and expected cost for meldable heaps.
Due to the more general implementation using a resourceful comparator and needing to deal with mutable state, the size of the \tachis{} formalisation has significant overhead with respect to our formalisation.
\citet{Matheja20} uses a probabilistic separation logic to give a pen-and-paper proof of the central theorem of \cite{GambinM98}, i.e., that the expected length of a chosen path in a tree is at most logarithmic in the total number of nodes.

\medskip
\verificationEffort{We claim that this case study demonstrates nearly full automation
since only two hints to the solver regarding properties of logarithms need to be provided for the verification of the optimal upper bound on the expected cost of \meld/.} 
\newpage
\subsection{Randomised Quicksort}
\label{sec:rand-quick}

\begin{figure}[t]
\begin{lstlisting}[style=haskell,numbers=left]
{-@ ~\DEF{rquick}~ :: Ord #a => {l:~\List/~ #a | ~\noDuplicates/~ l} -> ~\label{fig:quick-sort:rquick}~
        ~\ProperDist/~ {r:~\List/~ #a | sorted r && elems l = elems r} / [length l, 1] @-} ~\label{fig:quick-sort:rquick-sig}~
rquick ~\Nil/~ = ~\dunit/~ ~\Nil/~
rquick l = ~\uniform/~ 0 (length l - 1) (rquickBody l)

{-@ ~\DEF{rquickBody}~ :: Ord #a => {l:~\List/~ #a | ~\noDuplicates/~ l} -> {k:#N | k < length l} ->
        ~\ProperDist/~ {r:~\List/~ #a | sorted r && elems l = elems r} / [length l, 0] @-} ~\label{fig:quick-sort:rquickBody-sig}~
rquickBody l k = let e = ref l k in
                 let (smaller, bigger) = ~\partition/~ e (deleteAt l k) in ~\label{fig:quick-sort:rquickBody-def}~
                 ~\tick/~ (length l - 1) (~\combine/~ (rquick smaller) (rquick bigger) (merge e))

{-@ ref :: l:~\List/~ #a -> {k:#N | k < length l} -> #a @-}

{-@ deleteAt :: l:~\List/~ #a -> {k:#N | k < length l} -> {r:~\List/~ #a | 1 + length r = length l} @-}

{-@ ~\DEF{partition}~ :: Ord #a => x:#a -> {l:~\List/~ a | ~\noDuplicates/~ l && not (x ~$\in$~ l)} ->
        {t:(~\List/~ {v:#a | v < x}, ~\List/~ {v:#a | x < v})
          | length l = length (fst t) + length (snd t) &&
            elems l = Set_cup (elems (fst t)) (elems snd t)} @-}

{-@ merge :: Ord #a => x:#a -> {l1:~\List/~ {v:#a | v < x} | sorted l1} ->
    {l2:~\List/~ {v:#a | x < v} | sorted l2} ->
    {r:~\List/~ #a | sorted r && elems r = Set_cup (elems l1) (elems l2)} @-}
\end{lstlisting}

\caption{Randomised Quicksort}
\label{fig:quick-sort}
\end{figure}

Quicksort is a divide-and-conquer sorting algorithm that
partitions the input list into smaller parts, depending on a chosen pivot element, and recursively sorts the resulting parts.
While quicksort is fast for random input lists, the order of comparisons degrades to $\bO(n^2)$ in the worst case, if the choices of pivot elements are bad.
Randomised quicksort solves this problem by always choosing pivot elements uniformly at random from the input list, ensuring that the expected number of comparisons is in $\bO(n\log n)$ for any input.

Furthermore, the expected number of comparisons performed by randomised quicksort on a list with $n$ distinct elements is equal to $2(n+1)H_n - 4n$ \cite{EberlHN20}.
In additon to the functional correctness of the algorithm, which could be automatically deduced by \LH/, we were able to derive the closed-form solution for the expected number of comparisons by leveraging \LH/'s theorem proving capabilities.
In the following, we exemplify the implementation of the algorithm in the probability monad, after which the derivation of the expected number of comparisons is illustrated.
We end this case study by discussing related work.

\paragraph{Implementation.}
\Cref{fig:quick-sort} contains the implementation of the randomised quicksort algorithm using our probability monad.
We kindly note that the functional correctness is automatically verified by \LH/ from the implementation in combination with the specified refinement type annotations.
Thus, even though the usage of the monad adds a layer for expressing probabilistic functions, the refinement types are preserved by the provided monadic functions.

In line with the literature, we assume that all elements in the input list are unique as realised through the predicate \noDuplicates/.
Consider the refinement type of the function \rquick/ at line~\ref{fig:quick-sort:rquick}.
Given a list without duplicated elements over an ordered type \C!#a!, the function \rquick/ returns a proper distribution containing sorted lists whose set of elements is equal to the set of elements of the input list.
These properties are formalised using the measures \sorted/ and \elems/.
\LH/ requires termination metrics for all mutually recursive functions and thus also for \rquick/ and \rquickBody/ (see lines~\ref{fig:quick-sort:rquick-sig} and~\ref{fig:quick-sort:rquickBody-sig}, respectively).
We use a lexicographic termination metric, as in calls to \rquickBody/ the second component decreases from 1 to 0 and in calls to \rquick/ the first component, i.e., \C!length l! decreases.

Our implementation makes use of the probability monad as follows.
In the case that the input list is not empty, the algorithm samples an index from the list uniformly at random using the \uniform/ function and passes it to \rquickBody/ defined in line~\ref{fig:quick-sort:rquickBody-def}.
There, the expression \C!ref l k! computes the element at index \C!k! in list \C!l! and binds it to \C!e!.
Moreover, \C!deleteAt l k! evalutes to a list where the element at index \C!k! is removed.
The value \C!e! is used to partition the list \C!deleteAt l k! into a list of all the smaller values and a list with all the bigger values with respect to \C!e!.
Since \C!e! is compared with all the other elements, there are in total \C!length l - 1! comparisons to account for.
This is represented by the ticking functionality of the probability monad.
Then, the two lists are sorted recursively and by using the \combine/ function, the sorted parts are merged together with the help of the function \C!merge!.
The syntax \C!Set_cup! denotes the union of sets.

\paragraph{Expected Number of Comparisons.}

As clarified in ~\Cref{sec:overview:quicksort}, the formalisation establishing the closed form solution for the expected number of comparisons performed by \rquick/ is divided into two steps, where we make extensive use of theorems from our summation library to closely follow the pen-and-paper proof given in~\cite{cichon-quick}.

We subsequently explain in detail how the two steps described in the overview are carried out within our \LH/ formalisation.

\paragraph{Step 1: Extracting a recurrence for \rquick/.}

The formalisation proves by induction over the length of the input list \C!l! that \C!$\HEcostOpen$rquick l$]$! is equal to \C!qs_rec (length l)!.
Since theorems in \LH/ are functions that return the unit type refined with a logical predicate proved by the theorem, we use the termination metric \C![length l]! to encode an inductive proof.
That is, all recursive calls to the theorem must occur with smaller lists.

For an empty list it holds trivially that \C!$\HEcostOpen$rquick Nil$]$ = 0!, as \dunit/ incurs no costs.
This case is also automatically verified by \LH/ as the refinement type of \dunit/ states that the cost of the resulting distribution is zero.

Otherwise, if \C!l! is non-empty, the cost of \C!rquick l! can be easily determined by exploiting the refinement type of \uniform/.
Let \C!n! denote the length of \C!l!.
Remember, that the expected cost of \C!$\uniform/$ lo hi f! is given by \C!1/(hi-lo+1) * $\HfinSumTT{lo}{hi}$ ($\HEcostOnly$ $\circ$ f)!.
Thus, the cost of \C!rquick l! is given by \C!1/n * $\HfinSumTT{0}{n-1}$ ($\HEcostOnly$ $\circ$ rquickBody l)!.
The cost of \C!rquickBody l k! is equal to \C!rquickBodyCost!, as stipulated by the refinement on line~\ref{fig:quick-sort:rquickBody-sig} in \Cref{fig:quick-sort}.
The automatic verification of this equality is possible because of the cost-aware typings for \tick/ and \combine/, since for all proper distributions \C!d1! and \C!d2! we have that \C!$\HEcostOpen$tick d1$]$ = 1 + $\HEcostOpen$d1$]$! and \C!$\HEcostOpen$combine d1 d2 f$]$ = $\HEcostOpen$d1$]$ + $\HEcostOpen$d2$]$!.

%
%

In order to abstract away from the specific choices of pivot index, we use the notion of \emph{rank}.
That is, the rank of a value with respect to a list is the number of elements that are smaller than the respective value.

In \LH/ we realise this through a function \C!rank :: Ord #a => List #a -> #a -> #N! and use it to further refine the type of \partition/ as highlighted in green:

\begin{lstlisting}[style=haskell]
{-@ partition :: Ord #a => x:#a -> {l:List a | ~\noDuplicates/~ l && not (x ~$\in$~ l)} ->
        {t:(List {v:#a | v < x}, List {v:#a | x < v})
          | ~\greenBox{length (fst t) = rank l x \&\& length (snd t) = length l - 1 - rank l x}~ && $\color{gray}\ldots$} @-}
\end{lstlisting}

This additional annotation is automatically verified by \LH/ and can be employed together with applications of the induction hypothesis to replace the expected cost of \rquick/ to \C!qs_rec!.
That is, in the current state of the proof, we have shown that \C!$\HEcostOpen$rquick l$]$! is equal to
\begin{lstlisting}[style=haskell]
n-1 + 1/n * $\HfinSumTT{0}{n-1}$(\k -> qs_rec (rank l (ref l k)) + qs_rec (n-1-(rank l (ref l k))))
\end{lstlisting}
\newsavebox\natLess
\begin{lrbox}{\natLess}
\lstinline[style=haskell,basicstyle=\color{black}\footnotesize\ttfamily]!{v:Nat | v < c}!
\end{lrbox}

We introduce the function \C!idx_rank! that maps a given index to a rank by looking up the element and computing its rank.%
\footnote{We write $\mathbb{N}^{< \mathtt{c}}$ for the refinement type {} \usebox{\natLess}} 

\begin{lstlisting}[style=haskell]
{-@ idx_rank :: Ord #a => {l:List #a | ~\noDuplicates/~ l} -> $\mathbb{N}^{< \mathtt{length\ l}}$ -> $\mathbb{N}^{< \mathtt{length\ l}}$ @-}
idx_rank l k = rank l (ref l k)
\end{lstlisting}

Now, the proof employs the summation theorem \thmSumRewrite/ to rewrite the summation to the equivalent form depicted below.
\begin{lstlisting}[style=haskell]
n-1 + 1/n * $\HfinSumTT{0}{n-1}$ ((\k -> qs_rec k + qs_rec (n-1-k)) $\circ$ (idx_rank l))
\end{lstlisting}

The final step of the proof consists in removing the call to \C!idx_rank l! in the summation.
This is achieved by the theorem \thmSumPermute/ that allows to permute the sum indices.
To this end, we need an inverse of \C!idx_rank! and define the function \C!rank_idx! which converts the rank of a given element to its index by exploiting the fact that the rank equals the position of an element in an already sorted list and looking up the index of this element in the original list.

\begin{lstlisting}[style=haskell]
{-@ rank_idx :: Ord #a => {l:List #a | ~\noDuplicates/~ l} -> $\mathbb{N}^{< \mathtt{length\ l}}$ -> $\mathbb{N}^{< \mathtt{length\ l}}$ @-}
rank_idx l rk = idxOf l (ref (sorted l) rk)
\end{lstlisting}

Then, we establish that the functions \C!rank_idx l! and \C!idx_rank l! induce a permutation on {\footnotesize $\mathbb{N}^{< \mathtt{length\ l}}$} and prove the following theorem:
\begin{lstlisting}[style=haskell]
{-@ thm_idx_rank_bijection :: Ord #a => {l:List #a | ~\noDuplicates/~ l} -> x:$\mathbb{N}^{< \mathtt{length\ l}}$ ->
        { idx_rank l (rank_idx l x) = x && rank_idx l (idx_rank l x) = x } @-}
\end{lstlisting}
Note that it is crucial that \C!l! contains no duplicated elements since otherwise two elements may share the same rank.

Finally, with the help of \thmSumPermute/ the call to \C!idx_rank l! can be elimininated and we can conclude the formalised proof establishing the desired equality between \C!$\HEcostOpen$rquick l$]$! and \C!qs_rec (length l)!.

\paragraph{Step 2: Deriving a closed form of \C!qs_rec!.}

For the second step, we roughly follow the steps of the pen-and-paper proof given in \citet{cichon-quick}.
First, \C!qs_rec! is first simplified into the recurrence \C!qs_rec_simp!:

\begin{lstlisting}[style=haskell]
qs_rec_simp 0 = 0
qs_rec_simp n = n-1 + 2/n * $\HfinSumTT{0}{n-1}$ qs_rec_simp
\end{lstlisting}

The equivalence follows from splitting the sum in \C!qs_rec! using \thmSumLinear/ and reversing the order of the summation containing \C!qs_rec (n-1-k)! with the help of \thmSumReverse/.
We further transform the recurrence \C!qs_rec_simp! into a summation using the following theorem, which we use to exemplify the power of automation that \LH/ offers for theorem proving.

\begin{lstlisting}[style=haskell]
{-@ thm_q_rec_helper :: n:#N -> { q_rec_simp n = 2*(n+1) * $\HfinSumTT{1}{n}$ (\k -> (k-1) / (k*(k+1))) } @-}
thm_q_rec_helper 0 = ()
thm_q_rec_helper 1 = ()
thm_q_rec_helper n = $\HfinSumTT{0}{n-1}$ q_rec_simp ~\triRight/~ $\HfinSumTT{1}{n}$ (\k -> (k-1) / (k*(k+1))) ~\triRight/~ thm_q_rec_helper (n-1)
\end{lstlisting}

The theorem \C!thm_q_rec_helper! is proved by induction on the input \C!n!.
Note that \LH/ automatically tries to use the first argument of functions for termination checking and we do not need to specify the termination metric \C![n]! here.
The cases for \C!0! and \C!1! are verified without further input and thus it suffices to use the unit value \C!()!.
In the last line, it is known that \C!n! is greater or equal than \C!2! as the previous equations did not match.
Using the triangle operator \C!$\triRight/$! we add the relevant summation expressions.
This has the effect that \LH/ takes their refinement types into account and can unfold them as deeply as possible.
This is achieved by \LH/'s \emph{Proof By Logical Evaluation} (\textsf{PLE}) proof search algorithm \cite{ref-refl}.
Lastly, the theorem is called recursively, which corresponds to an invocation of the induction hypothesis.

With a similar theorem we establish the equality%
\begin{tightcenter}
\C!$\HfinSumTT{1}{n}$ (\k -> (k-1) / (k*(k+1))) = 2/(n+1) + (harmonic n) - 2!.
\end{tightcenter}
From which we can derive within \LH/ the final expression 
\C!2*(n+1) * (harmonic n) - 4*n! for the expected number of comparisons on randomised quicksort on a list of length \C!n!.

\paragraph{Comparison with Related Work}

\citet{EberlHN20} establish the same result about the expected number of comparisons of randomised quicksort in the proof assistant Isabelle/HOL by also exploiting the index-rank bijection.
While their formalisation also uses a monadic approach for representing probabilistic algorithms, there is no built-in support for reasoning about expected costs as distributions are represented as \emph{probability mass functions}, that is, functions with type $\alpha \to [0,1]$.
Instead, they use pairs of outcomes together with costs which demands that the cost for every monadic operation needs to be explictily passed around, whereas the
costs for our implementation are automatically tracked and can be inferred in a syntax-directed manner from the refinement types of the functions provided by the probability monad.
Our formal proof is approximately the same size as the one from \citet{EberlHN20}.
Moreover, because of the used representation of distributions in Isabelle/HOL, the formalised quicksort algorithm cannot be executed whereas we carried out the mechanised proof on a real-world implementation.

An upper bound on the number of comparisons is presented in \citet{WeegenMcKinna} using the \coq{} proof assistant and the existing real number theory of the \coq{} standard library.
Distributions are modelled as decision trees, where a tree is either a leaf node containing an outcome or a non-empty list of trees $[T_1, \ldots, T_n]$ with the meaning that the outcomes in subtree $T_i$ happen with probability $\sfrac{1}{n}$.
In order to represent probabilistic programs the decision trees are equipped with a monadic structure.
The development is based on the pen-and-paper proof given in \citet{Cormen:2009} which requires a reduction to expected pairwise comparison counts between elements of the input list.
Again, the size of the presented \LH/ formalisation of quicksort is on par.

A tail bound on the expected number of comparisons of randomised quicksort is mechanised in~\citet{Tassarotti018} using a monadic encoding using probability mass functions.

\citet{HaselwarterLMG024} introduce the higher-order separation logic \tachis{} to reason about the expected cost of probabilistic programs.
Randomised quicksort is mechanised as a case study therein, where an upper bound on the expected number of comparisons is derived whose 
proof comprises about twice the lines of our \LH/ development.

\medskip
\verificationEffort{
We believe this case study demonstrates that our monad in combination with \LH/'s automation allows for complex proof development while keeping up with formal proof developments from mature proof assistants.
}
\subsection{Randomised Splay Trees}
\label{sec:splay-trees}

In this case study we carry out an expected \emph{amortised} cost analysis of randomised splay trees \citet{ALBERS2002213} in which we verify the state-of-the-art bounds reported in \cite{LMZ:2022}.
Splay trees were first introduced by \citet{sleator1985self} where it was shown that the efficiency of the data structure relies on the restructuring heuristic called \emph{splaying} in which accessed elements are moved to the top of the tree using rotations.
While some accesses in a splay tree may be expensive, i.e., if elements stored in leaves need to be retrieved, whole sequences of operations performed on splay trees are cheap.
This concept is kown as \emph{amortised cost}.

In our case study we use the so-called \emph{physicist's method} performing amortised cost analysis where a \emph{potential function} $\phi$ is employed that maps the state of a datastructure $D$ into a real number, where $\phi(D)$ is called the \emph{potential} of $D$ \cite{tarjan1985amortized}.
If $c$ denotes the cost of some operation performed on $D$ and $D'$ is the state of $D$ after performing said operation, then $c + \phi(D') - \phi(D)$ is denoted the \emph{amortised time}.
In a probabilistic setting, performing an operation on $D$ results in a distribution of states.
Hence, for the amortised time one needs to replace the cost $c$ by the expected cost and $\phi(D')$ by the expected value of the resulting distribution with respect to expectation $\phi$.

The randomised splay trees proposed in \cite{ALBERS2002213} exploit the fact when splaying, it is not necessary to perform every rotation. 
Instead, the rotations are only performed with some fixed probability.
An automated expected amortised cost analysis of this data structure is performed in \cite{LMZ:2022}, which uses a first-order functional language with support for sampling over discrete distributions.
We translated the given implementation of \C!splay! into our cost-aware probability monad and were able to semi-automatically verify the same bound on the expected amortised time.
The cost model employed by the automated analysis uses a cost of $\sfrac{1}{2}$ for both rotations and recursive calls and a rotation is performed with a fixed probability of $\sfrac{1}{2}$.

While functional correctness is established fully automatically by \LH/, some hints regarding logarithmic reasoning and changes in potential resulting from rotations are needed for a successful verification for the upper bound of $\sfrac{9}{8} \log_2(|t|)$ on the amortised time for splaying performed on a tree $t$, where $|t|$ denotes its number of leaves.
For readability, we will also use the notation \C!|t|! to denote the size of tree \C!t! in listings.

The main result is given in the following annotations, where \C!treeSet! is a measure that computes the set of elements contained in the tree:
\begin{lstlisting}[style=haskell]
{-@ data BST #a = Empty | Node (x :: #a) (BST {v:#a | v < x}) (BST {v:#a | x < v}) @-}

{-@ splay :: Ord #a => #a -> t:BST #a ->
      {d:~\ProperDist/~ {r:BST #a | treeSet t = treeSet r && |t| = |r|}
        | $\HEcostTT{d}$ + $\HEval{\phi\Delta\,\,\mathtt{t}}{d}$ <= 9/8 * ~\logTwo/~ |t|} @-}
\end{lstlisting}
The distribution computed by \C!splay x t! contains binary search trees containing the same elements as the input tree \C!t!.
Moreover, the formalisation verifies the bound on the amortised cost by showing that the expected cost of \C!splay! plus the expected difference in the potential resulting from splaying satisfies the given upper bound.
That is, \C!$\phi\Delta$ :: BST #a -> BST #a -> #R! computes the difference of potentials, i.e., \C!$\phi\Delta$ t1 t2 = $\phi$ t1 - $\phi$ t2!, where we define $\phi$ as a reflected function in \LH/ as given in \cite{LMZ:2022}.

\begin{figure}
\begin{lstlisting}[style=haskell,numbers=left]
splay x t@(Node c cl@(Node b bl br) cr) =
  if x < c && x < b && nonEmpty bl then -- zig-zig case
    (~\logTwoIneq/~ |bl| (|br| + |cr|), ~\logTwoMonotone/~ |bl| |cl|) |> ~\label{lst:splay:log1}~
        coin${}^{\phi\Delta\,\mathtt{t}}$ 0.5
            (tick${}^{\phi\Delta\,\mathtt{t}}$ 0.5 (fMap${}^{\phi\Delta\,\mathtt{t}\, \leqslant\, \phi\Delta\,\mathtt{bl}\, +\, \mathtt{v}}$ (tick${}^{\phi\Delta\,\mathtt{bl}}$ 0.5 (splay x bl)) (zig_zig_rot t))) ~\label{lst:splay:rot}~
            (fMap${}^{\phi\Delta\,\mathtt{t}\, \leqslant\, \phi\Delta\,\mathtt{bl}\, +\, \mathtt{0}}$ (tick${}^{\phi\Delta\,\mathtt{bl}}$ 0.5 (splay x bl)) (zig_zig_norot t)) ~\label{lst:splay:norot}~
  where v = 3/4 * (~\logTwo/~ (|br| + |cr|) + ~\logTwo/~ |t| - ~\logTwo/~ |bl| - ~\logTwo/~ |cl|) ~\label{lst:splay:v}~
\end{lstlisting}

\caption{Zig-Zig case of \texttt{splay} function}
\label{fig:splay-zigzig}
\end{figure}

In order to keep track of the changes in potential caused by tree rotations, we introduce a variant of \fMap/ which supports keeping track of upper bounds on expected values.

\noindent\begin{minipage}{\linewidth}
\begin{lstlisting}[style=haskell]
{-@ fMapUB :: v:#R -> e1:(#b -> #R) -> e2:(#a -> #R) ->
      d:~\SubDist/~ #a -> f:(x:#a -> {v:#b | e1 v <= e2 x + v}) ->
      {r:~\SubDist/~ #b | ~\zsavepos{pos:fMapUBA}~r = ~\fMap/~ d f && ~\probMass/~ r = ~\probMass/~ d && $\HEcostTT{r}$ = $\HEcostTT{d}$ &&
~\zsavepos{pos:fMapUBB}\hspace{\inteval{\zposx{pos:fMapUBA}-\zposx{pos:fMapUBB}}sp}~$\HEvalTT{e1}{r}$ <= $\HEvalTT{e2}{d}$ + (~\probMass/~ d) * v} @-}
\end{lstlisting}
\end{minipage}

Note that it is automatically verified that the results of \C!fMapUB! and \fMap/ coincide.
Additionally, the refinement type of \C!fMapUB! is automatically verified by \LH/ and its implementation performs just a simple recursion on the given distribution.

The function \C!fMapUB! can now be used in the verification for \C!splay! by employing potential differences as expectations.
Consider a call \C!splay x t!.
Since we are interested in the expected potential differences between tree \C!t! and the trees in the resulting distribution, we employ the partially applied function \C!$\phi\Delta\,$t :: BST #a -> #R! as an expectation.
By supplying such an expectation to the functions of the probability monad, we can exploit the cost bound given by the recursive calls to \C!splay!.

One case of the splaying implementation is given in~\Cref{fig:splay-zigzig}.
For readability we write the expectation parameter passed to functions of the probability monad as a superscript and we write \C!fMap${}^{\mathtt{e1} \leqslant \mathtt{e2} + \mathtt{v}}$! for \C!fMapUB v e1 e2!.

In line~\ref{lst:splay:log1} we state two logarithmic facts required for the verification.
The rotation is performed in line~\ref{lst:splay:rot} with a probability of $\sfrac{1}{2}$ whereas in line~\ref{lst:splay:norot} no rotation is performed, which is implemented by the functions \C!zig_zig_rot! and \C!zig_zig_norot! respecively.
These functions specify a bound on the potential difference with respect to the tree \C!t! and the subtree resulting from the recursive splay operation which is then exploited by \C!fMapUB!.
In the case that no rotation is performed, potential remains unchanged and if a rotation is performed, the constant \C!v! given in line~\ref{lst:splay:v} is required.

\paragraph{Related Work}

The prototype implementation \atlas{} introduced in \cite{LMZ:2022} is able to perform a fully automated analysis for the expected amortised cost of splay trees and similar probabilistic data structures.
While our formalisation requires several hints for \LH/ regarding changes in potential and logarithmic facts, the functional correctness is fully automatically verified in contrast to \atlas/.
\citet{NipkowB19} develop a framework for amortised analysis within \isabelle{} wherewith a deterministic splay function is verified.
However, the framework lacks support for probabilistic algorithms.
In \citet{vanBrgge2024} \LH/ is used to carry out amortised analyses of binomial heaps and finger trees.
While logarithmic reasoning was also required, it sufficed to only implement $\lfloor \log_2(\cdot)\rfloor$ using iterated division by 2.
This approach is not feasible for our use case as we require exact properties of the real mathematical logarithm.
Likewise, probabilistic programs were out of scope in this work.

\medskip
\verificationEffort{
The formalisation of randomised splay trees uses the probability monad's support for ghost parameters to manually reason about changes to the potential of trees caused by rotations in combination invocations of necessary logarithmic facts.
For future work, we believe that the automatic instantiation of ghost parameters provides a higher level of automation.
Nonetheless, the combinators of the probability monad proved to be powerful enough to develop a formalisation without the use theorems, cf.~\Cref{table:loc}.
}
\subsection{Random Permutations}
\label{sec:random-perm}

\begin{figure}[t]
\begin{lstlisting}[style=haskell,numbers=left]
{-@ ~\DEF{permute}~ :: l:~\List/~ #a -> ProperDist ({l':~\List/~ #a | elems l = elems l'}) / [length l, 1] @-}
permute ~\Nil/~ = ~\dunit/~ ~\Nil/~
permute l = ~\uniform/~ 0 (length l - 1) (~\permuteBody/~ l)

{-@ ~\DEF{permuteBody}~ :: l:~\List/~ #a -> {i:#N | i < length l} ->
        ProperDist ({l':~\List/~ #a | head l' = ref l i && elems l = elems l'}) / [length l, 0] @-} ~\label{fig:rand-perm:permuteBody-annot}~
permuteBody l i = ~\fMap/~ (~\permute/~ (deleteAt l i)) (cons (ref l i))

{-@ cons :: x:#a -> l:~\List/~ #a ->
      {r:~\List/~ #a | head r = x && length r = length l + 1 && elems r = Set_add x (elems l)} @-}
\end{lstlisting}

\caption{Computing uniform random permutations}
\label{fig:rand-perm}
\end{figure}

A common technique to avoid unfavourable inputs is to randomly permute a given input list.
The goal of this case study is to showcase how our probability monad can be used to implement and verify an algorithm which generates random permutations.
In the remainder of this section, we first explain the implementation of the algorithm before its verification using \LH/'s theorem proving capabilities is demonstrated.

\paragraph{Implementation}

Given a list \C!l!, the function \permute/ (see~\Cref{fig:rand-perm}) uses the probability monad to generate a proper distribution of lists which are permutations of \C!l!.
Note that the refinemenet type annotations in line~\ref{lst:permute} express that for every list \C!l'! in the resulting distribution, the set of elements of \C!l'! equals the set of elements of \C!l!.
Similar to the randomised quicksort case study, we use \uniform/ to sample a list index which is passed to \permuteBody/.

Given a list \C!l! and index \C!i! such that \C!i < length l! the function \permuteBody/ uses \deleteAt/ to compute the list resulting from removing the \C!i!-th element of \C!l!.
The recursive call \C!$\permute/$ (deleteAt l i)! evaluates to the distribution of permutations of \C!deleteAt l i! and it remains to add the removed element back to the resulting lists.
This is achieved through \fMap/ in combination with \C!cons (ref l i)!, i.e., the element at index \C!i! in list \C!l! is prepended to each list in the distribution computed by \C!$\permute/$ (deleteAt l i)!.
Note that the refinement type of \C!permuteBody l i! guarantees that for every list \C!l'! in the resulting distribution, the first element of the list, i.e., \C!head l'!, is equal to the \C!i!-th element of \C!l!.
This property will be exploited in the case study formalising the hiring problem presented in~\Cref{sec:hiring}.

\paragraph{Verification}

Even though the refinement type annotations of \permute/ state that the element set of the lists contained in the computed distribution equal the set of elements of the given input list, this property does not suffice to show that we in fact generate random permutations.
Consider, for instance, a simple definition such as \C!$\permute/$ l = dunit l! which also satisfies the aforementioned refinement type.
To this end, we formalise within \LH/ that \permute/ generates \emph{uniform} random permutations, that is, the probability of every permutation of the input list in the resulting distribution is exactly the same.
To be more precise, we develop the \LH/ theorem \C!thm_permute_uniform! which, given lists \C!l1! and \C!l2! of the same length \C!n! without duplicated elements, proves that the probability of \C!l2! in the distribution \C!$\permute/$ l1! is equal to \C!1 / (1 * 2 * ... * n)!, as shown below
%
\begin{lstlisting}[style=haskell]
{-@ thm_permute_uniform :: {l1:~\List/~ #a | ~\noDuplicates/~ l1} ->
        {l2:~\List/~ #a | ~\noDuplicates/~ l2 && elems l1 = elems l2 && length l1 = length l2} ->
        { ~\probOf/~ l2 (~\permute/~ l1) = 1 / (~\factorial/~ (length l1)) } @-} ~\label{thm_permute_uniform:equality}
\end{lstlisting}

The formalised proof proceeds by pattern matching on the shape of \C!l1! and \C!l2!.
With the help of the totality checker of \LH/ we only have to consider two cases, where either both lists are empty, i.e., equal to \C!Nil! or both lists are nonempty.
Hence, as the refinement types constrain the shape of \C!l1! and \C!l2!, we do not have to take care of cases like \C!l1 = Nil! and \C!l2 = Cons x xs!.

The case where both lists are empty is automatically verified as \LH/'s \textsf{PLE} algorithm simply reduces both sides of the equality in line~\ref{thm_permute_uniform:equality} of~\Cref{fig:rand-perm} to \C!1!.
Otherwise \C!l1! and \C!l2! are of the form \C!Cons x1 x1s! and \C!Cons x2 x2s!, respectively and let \C!n! denote \C!length l1!.
We employ a helper theorem, which allows to rewrite \C!$\probOf/$ l2 ($\uniform/$ 0 (n-1) ($\permuteBody/$ l1))! into the equivalent expression \C!(1 / n) * $\HfinSumTT{0}{n-1}$ (probOf l2 $\circ$ $\permuteBody/$ l1)!, i.e., \uniform/ admits the already known sum representation for expected costs and expected values also for the probability of outcomes.
For ease of notation, let \C!f! subsequently denote \C!$\probOf/$ l2 $\circ$ $\permuteBody/$ l1!.

Our formalisation relies on a the helper lemma \C!thm_permute_uniform_lemma! establishing the equality
\begin{lstlisting}[style=haskell]
 $\HfinSumTT{0}{n-1}$ f = ~\probOf/~ x2s (~\permute/~ (deleteAt l1 (indexOf l1 x2))) ~\labelExp{expr:perm:1}~
\end{lstlisting}
As this lemma peforms the heavy lifting of our formalisation, we want to build some intuition for this statement.
First, \C!deleteAt l i! removes the element at index \C!i! from \C!l! and \C!indexOf l x! returns the index of the element \C!x! in list \C!l!.
The key insight is that given an index \C!i!, the value of the expression \C!$\probOf/$ l2 ($\permuteBody/$ l1 i)! can only be non-zero if \C!ref l1 i! equals \C!x2!, i.e., the head of the list \C!l2!, since the list head of every outcome in \C!$\permuteBody/$ l1 i! would be different otherwise, as captured by the refinement type annotations in line~\ref{fig:rand-perm:permuteBody-annot} of~\Cref{fig:rand-perm}.
In the following, let \C!j! denote \C!indexOf l1 x2!.
The proof splits the summation at the left-hand side of equation~\ref{expr:perm:1} at index \C!j! by utilising the sum theorem \thmSumSplit/, yielding the equivalent expression
\C!$\HfinSumTT{0}{j-1}\,$f + $\probOf/$ l2 ($\permuteBody/$ l1 j) + $\HfinSumTT{j+1}{n-1}\,$f!.
By the above reasoning \thmSumConstant/ is used to show that both \C!$\HfinSumTT{0}{j-1}\,$f! and \C!$\HfinSumTT{j+1}{n-1}\,$f! are zero.

With this result in place, we can finally finish the mechanised proof of \C!thm_permute_uniform!.
Remember, that \C!$\probOf/$ l2 ($\uniform/$ 0 (n-1) ($\permuteBody/$ l1))! was rewritten into the equivalent form \C!(1 / n) * $\HfinSumTT{0}{n-1}$ f!.
Now, a simple call to \C!thm_permute_uniform_lemma! is enough for \LH/ to deduce the equivalence
\begin{lstlisting}[style=haskell]
(1 / n) * $\HfinSumTT{0}{n-1}$ f = (1 / n) * ~\probOf/~ x2s (~\permute/~ (deleteAt l1 j))
\end{lstlisting}
Finally, by applying the induction hypothesis, i.e., calling the theorem \C!thm_permute_uniform! recursively with the smaller list \C!deleteAt l1 j! of length \C!n-1! the proof can be automatically finished, as the \textsf{PLE} algorithm can easily verify the equality between \C!(1 / n) * (1 / factorial (n-1))! and \C!1 / (factorial n)!.

\paragraph{Related Work}

A textbook proof for an imperative permutation generation program is given in \citet{Cormen:2009}.
The proof relies on a complicated loop invariant with an intricate argument why the invariant is preserved by the program loop depending on reasoning about conditional probabilities.
In \citet{Fisher_Yates-AFP} the correctness of the Fisher-Yates shuffling algorithm is proved.
This algorithm is conceptually equal to \permute/, as the algorithm traverses the input list and at each step swaps the current element with a random element from the remaining list.
The main theorem establishing the correctness of the Fisher-Yates shuffling algorithm in \cite{Fisher_Yates-AFP} needs 61 lines in Isabelle/HOL, slightly longer than our formalisation.

\citet{HaselwarterLMG024} use \tachis{} to verify an upper bound on the so-called \emph{expected entropy cost} of the Fisher-Yates shuffling algorithm, in which a uniform sample from a set of $N$ elements incurs a cost of $\log_2(N)$.

\medskip
\verificationEffort{
While the verification of \permute/ requires some \LH/ theorems, the proof structure can be extracted in a syntax-directed manner.
For example, the probability of \C!x! in \C!uniform lo hi f! is given by the weighted average of the probability of \C!x! in the distributions \C!f lo!, $\ldots$, \C!f hi!.
That is, the proof proceeds from first principles and requires no arguments depending on conditional probabilities in contrast to \cite{Cormen:2009}.
}    
\subsection{Hiring Problem}
\label{sec:hiring}

This case study formalises the \emph{hiring problem} as presented in \citet{Cormen:2009}.
Assume that a new person has to be hired for an open position and for each of $n$ consecutive days a candidate is sent by an employment agency for an interview.
There is a total order on the level of qualification of candidates, i.e., it can always be determined if one is more qualified than another.
An applicant is hired if the candidate is more qualified with respect to the presently employed person.
As the hiring and training of new employees is expensive, one wishes to estimate the number of times a new person is hired if $n$ candidates are considered.

Evidently, in the worst case all $n$ candidates are hired if they are interviewed in increasing order by qualification.
However, we are interested in the average case, i.e., the order of qualification is distributed uniformly at random.
To this end, we can make use of the already verified \permute/ function.
The hiring problem is formalised using the following code depicted in~\Cref{fig:rand-hiring}.

\begin{figure}[t]
\begin{lstlisting}[style=haskell,numbers=left]
{-@ ~\DEF{countHires}~ :: ~\List/~ #N -> #N @-}
countHires ~\Nil/~ = 0
countHires (~\Cons/~ x xs) = (if x > listMax xs then 1 else 0) + countHires xs

{-@ listMax :: ~\List/~ #N -> #N @-}
listMax ~\Nil/~ = 0
listMax (~\Cons/~ x xs) = max x (listMax xs)

{-@ ~\DEF{thm_expectVal_hireAssistant}~ :: {l:~\List/~ {v:#Z | v > 1} | noDuplicates l} ->
    { $\HEvalTT{\countHires/}{\permute/\ l}$ = harmonic (length l) } / [length l] @-}
\end{lstlisting}
\caption{Hiring Problem Formalisation}
\label{fig:rand-hiring}
\end{figure}

\paragraph{Implementation}
The function \countHires/ takes as input a list of natural numbers and computes the number of times that a new maximum value is encountered while scanning the list, representing the hiring of a better qualified candidate.
By using \countHires/ as an expectation, the theorem \thmExpectValHireAssistant/ formalises the statement of the hiring problem.
The given list is constrained to not contain duplicated values in order to model the total order on the qualification of the candidates.
In the following, we will detail the workings of the formalised proof of \thmExpectValHireAssistant/ as we used a different approach than the pen-and-paper proof presented in \citet{Cormen:2009}.

\paragraph{Verification}
As the termination metric of the theorem implies, we perform an inductive proof on the length of the input list \C!l!, where the base case is fully automatically verified by \LH/.
Now, consider a non-empty list \C!l! satisfying the needed refinement type as stipulated by \thmExpectValHireAssistant/ and let \C!n! denote \C!length l!.
As in previous case studies, the expected value of an expression involing \uniform/ can be rewritten into a summation over the expected values of the resulting distributions that depend on the uniformly chosen value.
In our case the summation is given by
\begin{lstlisting}[style=haskell]
(1 / n) * $\HfinSumTT{0}{n-1}$ ($\HEvalTTOnly{\countHires/}$ $\circ$ ~\permuteBody/~ l) ~\labelExp{expr:hiring:1}~
\end{lstlisting}

The next step of the proof consists in unfolding the function \countHires/ for each summation term.
More concretely, for a non-empty list \countHires/ checks if the list head is greater than the maximum of the tail of the list and increments the amount of hired candidates if this is the case.
Morever, as stipulated by the refinement for \C!$\permuteBody/$ l i!, the list head for every list of the resulting distribution is equal to \C!ref l i! while the tail of the list is given by \C!$\permute/$ (deleteAt l i)!.
With the help of \thmSumLinear/ and by introducing the auxiliary expectation \C!f!, where \C!f i l' = if ref l i > listMax l' then 1 else 0!, expression~\ref{expr:hiring:1} is rewritten into
\begin{lstlisting}[style=haskell,xleftmargin=0pt]
(1/n) * $\Bigl(\HfinSumTT{0}{n-1}$($\HEvalTTOnly{\countHires/}$ $\circ$ ~\permute/~ $\circ$ deleteAt l) + $\HfinSumTT{0}{n-1}$($\lambda$i $\to$ $\HEvalTTOnly{f\,i}[$~\permute/~ (deleteAt l i)$]$)$\Bigr)$ ~\labelExp{expr:hiring:2}~
\end{lstlisting}

The proof of \thmExpectValHireAssistant/ is completed by invoking the induction hypothesis for each summand of the left summation in expression~\ref{expr:hiring:2}, i.e., the summation can be replaced by \C!n * harmonic (n-1)!.
Additionally, as the maximum of lists is preserved under permutations, the right summation of expression~\ref{expr:hiring:2} can be replaced by
\begin{lstlisting}[style=haskell]
$\HfinSumTT{0}{n-1}$(\i -> if ref l i > listMax (deleteAt l i) then 1 else 0)
\end{lstlisting}
The value of this summation is equal to \C!1! as only the index of the maximum value of the list \C!l! causes the expression to become equal to \C!1!.
Hence, \LH/ then is able to show the equality between \C!(1/n) * (n * harmonic (n-1) + (1/n))! and \C!harmonic n!.


\paragraph{Related Work}

\citet{Cormen:2009} solve the hiring problem through a high-level proof using indicator variables.
While the textbook includes pseudocode that models the hiring problem, the proof does not refer to the code and uses assumptions such as that candidate $i$ is better qualified than candidates 1 through $1-i$ with a probability of $1/i$.
In contrast, our formal proof is based on the already verified \permute/ functionality and leverages the support of the probability monad for reasoning with expectations in combination with the necessary summation manipulation theorems.

\medskip
\verificationEffort{
Similar to the verication of uniform random permutations in~\Cref{sec:random-perm}, the overall proof structure can be deduced by the implementation of the respective functions.
Some care is needed to exploit the fact that the head element of each list contained in the distribution computed by \C!$\permuteBody/$ l i! is equal to \C!ref l i! in combination with unfolding \countHires/ to obtain the two summations highlighted in expression~\ref{expr:hiring:2}.
As can be seen in~\Cref{table:loc}, the size of the formal verification is considerably larger than the one carried out in~\Cref{sec:random-perm}.
} 

\section{Correctness of Static Probabilistic Cost Analysis}
\label{sec:soundness}

In this section we state and prove the correctness of cost analysis for probabilistic programs in \LH/.
We follow the set-up of the correctness proof for cost analysis as stated in~\cite{HandleyVH20}.

\subsection{Metatheory of \LH/}
\label{subsec:metatheory_-LH}

\begin{figure}
\def\arraystretch{1.2}
\begin{tabular}{rccll}
  Constants & $c$ & $\defsymb$ & $n \in \mathbb{Z} \mid s \in \mathbb{R} \mid \mathit{true},\mathit{false} \in \mathbb{B}$  \\
  & & | & $+,\ -,\ \ldots \mid \ =,\ <,\ \ldots \mid \land,\ \lor,\ \ldots$  \\
  & & | & \greenBox{$\mathit{dunit},\mathit{bind},\mathit{tick}, \mathit{fMap}, \mathit{coin}, \mathit{uniform}$} \\
  Values & $v$ & $\defsymb$ & $c \mid \lambda x.e \mid D \ \overline{e}$  \\
  Expressions & $e$ & $\defsymb$ & $v \mid x \mid e \ e \mid \mathit{let} \ x = e \ \mathit{in} \ e \mid \mathit{case} \ e \ \mathit{of} \{D \ \overline{y} \rightarrow e\}$  \\
  Refinements & $r$ & $\defsymb$ & $e$  \\
  Basic types & $B$ & $\defsymb$ & $\mathbb{Z} \mid \mathbb{R} \mid \mathbb{B} \mid \mathit{T}$  \\
  Types & $\tau$ & $\defsymb$ & $\{x:B \mid r\} \ \mid x:\tau_x \rightarrow \tau \mid$ \greenBox{$\mathit{SubDist} \ \tau$}\\
  Evaluation Contexts & $\evalContextSymb$ & $\defsymb$ & $\bullet \mid \evalContextSymb \ e \mid c \ \evalContextSymb \mid \mathit{case} \ \evalContextSymb \ \mathit{of} \{D \ \overline{y} \rightarrow e\}$  \\
\end{tabular}
\caption{Syntax of $\lcName$ (core \LH/ plus our probabilistic extensions highlighted in \greenBox{green}) }
\label{fig:meta}
\end{figure}

\begin{figure}
\[
\def\arraystretch{1.1}
\begin{array}{rcl}
\evalContext{e} & \smallStep & \evalContext{e'} \ \ \ \ \text{if} \ e \smallStep e' \\
c \ v & \smallStep & \delta(c,v) \\
(\lambda x.e) \ e_x& \smallStep & e[e_x/x]
\end{array} \quad \begin{array}{rcl}
\mathit{let} \ x = e_x \ \mathit{in} \ e & \smallStep & e[(\mathit{fix} \ \lambda x.e_x)/x] \\
\mathit{fix} \ e & \smallStep & e (\mathit{fix} \ e) \\
\mathit{case} \ D_j \ \overline{e} \ \mathit{of} \{D_i \ \overline{y_i} \rightarrow e_i\} & \smallStep & e_i[\overline{e}/ \overline{y_i}]
\end{array}
\]
\caption{Small-step Operational Semantics of $\lcName$ (without our probabilistic extensions)}
\label{fig:small-step-semantics}
\end{figure}
In this subsection we review the metatheory of \LH/, compare~\cite{HandleyVH20}.
In \Cref{fig:meta} we state the syntax and semantics of $\lcName$, a language that models the core of \LH/, where we have marked the parts that concern our probabilistic extensions, which we are going to discuss in the next subsection.
The language $\lcName$ includes constants, abstractions, applications, recursive definitions ($\mathit{let} \ x = e \ \mathit{in} \ e$), case statements, and datatypes.
We require all recursive definitions to be terminating; our logic is agnostic to the mechanism used to enforce structural
termination, so we leave this mechanism abstract.
In our implementation we rely on the termination checker of \LH/, which either is able to automatically prove termination based on structural metrics or which can verify the correctness of a user-provided termination metric.
The \emph{operational semantics} of $\lcName$ is defined in \Cref{fig:small-step-semantics} as a contextual, small-step, call-by-name relation $\smallStep$ whose reflective, transitive closure is denoted by $\smallStep^*$.

\paragraph{Constants.}
Constants applied to values are reduced in one step using the primitive constant
operation $c \ v \smallStep \delta(c,v)$.
For example, consider $(+)$, the primitive addition operator on integers.
In this instance, $\delta(+,n) = +_n$, where $+_n$ is the function that takes some integer $m$ and returns $n+m$.

\paragraph{Types.}
The \emph{basic types} in $\lcName$ are integers ($\mathbb{Z}$), reals ($\mathbb{R}$), booleans ($\mathbb{B}$) and type constructors (represented by the symbol $\mathit{T}$).
\emph{Types} are either \emph{refinement types} of the form $\{x:B \mid r\}$ where the basic type $B$, captured by the variable $x$, is refined by the boolean expression $r$; or \emph{dependent function types} of the form $x:\tau_x \rightarrow \tau$, where the input $x$ has the type $\tau_x$ and the result type $\tau$ may refer to the binder $x$.

\paragraph{Denotations.}
Each type $\tau$ denotes a set of expressions $\denotOf{\tau}$, defined by the dynamic semantics in~\cite{conf/icfp/VazouSJVJ14}.
Let $\remRef{\tau}$ the type obtained by erasing all refinements from $\tau$ and $e:\remRef{\tau}$ be the standard typing relation for the $\lambda$-calculus.
Then, we define the denotation of types as follows:
\[
\begin{array}{lcl}
  \denotOf{\{x:B \mid r\}} & \doteq & \{e \mid e:B, \text{ if } e \smallStep^* v, \text{ then } \mid r[v/x] \smallStep^*  \mathit{true} \} \\
  \denotOf{x:\tau_x \rightarrow \tau} & \doteq & \{e \mid e:\remRef{\tau_x \rightarrow \tau}, \text{ for all } e_x \in \denotOf{\tau_x} \text{ it holds that } e \ e_x \in \denotOf{\tau[e_x/x]} \}
\end{array}
\]

\paragraph{Syntactic Typing}
The typing judgement $\Gamma \vdash e:\tau$ decides syntactically if e is a member of $\tau$'s denotation using the environment $\Gamma$ that maps variables to their types:
$\Gamma \doteq x_1:\tau_1, \ldots, x_n:\tau_n$.

\paragraph{Typing of Constants}
To type a $\lcName$ constant $c$, we use the meta-function $\mathit{Ty}(c)$ that returns the type of $c$, which is used in the typing axiom \( \overline{\Gamma \vdash c:\mathit{Ty}(c)} \).
To ensure soundness, $\mathit{Ty}(c)$ should satisfy denotational inclusion: $c \in \mathit{Ty}(c)$.


\begin{theorem}[Soundness of Core \LH/~\cite{conf/icfp/VazouSJVJ14}]
If for all constants $c$, $c \in \mathit{Ty}(c)$, then $\emptyset \vdash e : \tau$ implies $e \in \denotOf{\tau}$.
\end{theorem}

\subsection{Probabilistic Cost Analysis}

We now state the full definition of $\lcName$, i.e., we give definitions to our probabilistic extensions.

\paragraph{Extensions of Types and Constants.}
We define the type $\mathit{SubDist}$ as well as the constants
$\mathit{dunit}$, $\mathit{bind}$, $\mathit{tick}$, $\mathit{fMap}$, $\mathit{coin}$, and $\mathit{uniform}$ as in \Cref{sec:prob-monad}.

Moreover, the functions $\mathit{expectCost}$, $\mathit{expectVal}$, $\mathit{probOf}$, and $\mathit{probMass}$ are provided for use in refinement type annotations.
Note that this means that we define our extension in terms of core \LH/.
We further point out that while we have introduced $\mathit{SubDist}$ in terms of sequences of outcomes, we treat them as abstract types in $\lcName$, e.g., we do not allow applying functions on sequences of outcomes of these types.
Rather, the only way to manipulate these types is via the operations $\mathit{dunit},\mathit{bind},\mathit{tick}, \mathit{fMap}, \mathit{coin}$, and $\mathit{uniform}$.

\paragraph{Semantics and Soundness of Constants.}
We stress that
as we use core \LH/ terms to define the constants $\mathit{dunit},\mathit{bind},\mathit{tick}, \mathit{fMap}, \mathit{coin}$, and $\mathit{uniform}$,
their semantics is induced by the semantics for the core \LH/ terms.
Further, we can rely on the (refinement) type checker of \LH/ to check the (refinement) types of our definitions of these constants,
i.e., we have $c \in \mathit{Ty}(c)$ for each of the constants $\mathit{dunit},\mathit{bind},\mathit{tick}, \mathit{fMap}, \mathit{coin}$, and $\mathit{uniform}$.
Therefore, these constants can be \emph{safely} used in $\lcName$ while preserving soundness.

We now obtain the soundness of our probabilistic cost monad as a corollary to the soundness of core \LH/:

\begin{theorem}[Soundness of Probabilistic Cost Analysis]
\label{thm:soundness}
Let $p: \mathbb{R} \rightarrow \mathbb{B}$ and $q: \mathbb{R} \rightarrow \mathbb{B}$ be predicates, let $f:\tau\rightarrow \mathbb{R}$ be an expectation, and let $e: \{d : \mathit{SubDist} \ \tau \mid p \ (\mathit{expectCost} \ d) \land q \ (\mathit{expectVal} \ f \ d) \}$ be an expression with $e\smallStep^* v$.

Then, $p \ (\mathit{expectCost} \ v) \smallStep^* \mathit{true}$ and $q \ (\mathit{expectVal} \ f \ v) \smallStep^* \mathit{true}$, as well as $w \in \denotOf{\tau}$ for all $w$ contained in $\mathit{support}\ v$.
\end{theorem}

\section{Related Work}
\label{sec:related-work}

Essentially starting with the seminal work of Kozen~\cite{Kozen81,Kozen:JCSC:85}, there is a large body of work on the analysis of \emph{probabilistic} programs (see also~\cite{BKS:2020} for further pointers).
Similarly, the literature on the verification of \emph{functional} programs
is extensive (see for example~\cite{EberlHN20} and the references therein). Furthermore, there is
a significant number of works on the expected (amortised) cost analysis of probabilistic functional programs,
for example~\cite{NipkowB19,HandleyVH20,WangKH20,LMZ:2021,LMZ:2022}. For brevity, we thus restrict to
closely related work.

\paragraph{Probabilistic reasoning}

Typically, probabilistic reasoning is formalised in functional languages
like Haskell through \emph{probability monads}~\cite{Ramsey02, Erwig06, Scibior15, TassarottiH19}.
Similarly, relational properties---in particular coupling---can be encoded via a monadic
implementation, cf.~\citet{VasilenkoVB22}. \citet{VasilenkoVB22} achive this in a similar
fashion as our contribution by enhancing \LH/ through a dedicated probability monad.
In contrast to our cost-aware probability monad no seamless reasoning about costs is supported.
Moreover, the \LH/ code of \cite{VasilenkoVB22} relies on several axioms to establish relational axioms as well as a notion of distance between distributions.
The only axioms used in our work are three basic properties of logarithms, which is a necessary choice in order to support reasoning about the real mathematical logarithm.
Additionally, the case studies performed to evaluate the methodology in~\cite{VasilenkoVB22} focus on
verifying classical machine learning properties, like \emph{convergence} and \emph{stability}. Our motivation
stems from the analysis of non-functional program properties, to wit, the \emph{expected costs} of
randomised algorithms.

\paragraph{Expected cost analysis}

The (automated) cost analysis of probabilistic programs has been intensively studied
(see, for example~\cite{KKMO:ACM:18,NgoCH18,NipkowB19,AMS20,EberlHN20,HandleyVH20,KKM20,WangKH20,ABD:2021,MoosbruggerBKK21,LMZ:2022,AvanziniMS23,AvanziniBGMV24,HaselwarterLMG024,ChatterjeeGMZ24,ZaiserMO25}).
Martingale based techniques have been implemented, e.g.,
by Peixin~Wang et al.~\cite{WFGCQS:PLDI:19}. Related results have been reported by
\citet{MoosbruggerBKK21}.
In the following, we discuss in more details the literature on cost analysis of probabilistic
functional programs.

\citet{WangKH20} establish automated support for the analysis of probabilistic higher-order (functional) programs, implemented in the automated resource analysis tool~\raml~\cite{HoffmannDW17}. The analysis not only provides an expected cost analysis, but also an expected \emph{amortised} analysis. 
In a similar vain~\citet{LMZ:2022} establish the prototype implementation~\atlas, that automatically infers precise expected (amortised) cost bounds for various functional data structures.
As the focus of these tools is on automated inference of (precise) cost bounds, the programs in scope are necessarily restricted. Nevertheless, for two of our case studies---\emph{meldable heaps}~\cite{GambinM98} and \emph{randomised splay trees}~\cite{ALBERS2002213}---%
\citet{LMZ:2022} report the automated inference of precise (and optimal) expected cost bounds.
On the other hand, clearly, our extension of a refinement type system is much more
generally applicable for verification. In particular, the approach of \citet{LMZ:2022}
cannot handle---not even express---our remaining case studies.

Recently \citet{HaselwarterLMG024} have established \tachis, a higher-order seperation logic,
formulated for a probabilistic higher-order functional programs. \tachis\ constitutes a powerful program logic, cleverly using \emph{credits} to represent expected costs and other cost metrics like the consumed entropy. Further, similar to the work in~\cite{WangKH20,LMZ:2022}, the analysis can establish expected \emph{amortised} cost bounds. A full~\coq\ implementation of~\tachis\ is available.%
\footnote{See~\url{https://zenodo.org/records/12659527}.}
Methodologically, \tachis\ is orthogonal to the here developed probability monad.
\tachis\ aims to be a rich reasoning tool for higher-order functional programs,
while our aim is to ease the formalisation burden as much as possible,
fascilitating automated support whenever possible. Like in our work~\citet{HaselwarterLMG024} include meldable heaps and randomised quicksort in their case studies.
As depicted in~\Cref{table:loc-comp} our approach leads to significant lower formalisation
efforts. The formalisation of meldable heaps requires even a magnitude less code than
the corresponding encoding in~\tachis. With respect to randomised quicksort the formalisation
in our cost-aware probability monad requires half the code than the corresponding formalisation in~\tachis.

\paragraph{Randomised algorithms and data structures}

Apart from meldable heaps, randomised splay trees and randomised quicksort, we have
also considered \emph{randomised quickselect}~\cite{Hoare61a,devroye1984exponential} (see the Appendix),
\emph{random permutations}~\cite{durstenfeld1964algorithm} and the
\emph{hiring problem}~\cite{ajtai2001improved} (see the Appendix) as case studies.

As already mentioned, randomised quicksort and meldable heaps
have been considered by~\citet{HaselwarterLMG024}. Further formalisations were
done in \coq, cf.~\citet{WeegenMcKinna} and in \isabelle, cf.~\citet{EberlHN20}.
Randomised quickselect was studied in~\cite{AvanziniBGMV24}.
Meldable heaps, as well as randomised splay trees have been studied by~\citet{LMZ:2022},
even in a fully automated way. To the best of our knowledge neither random permutations
nor the hiring problem as been formalised before.

The comparison to related work for all our case studies, with the exception of
the hiring problem and randomised quickselect, are given in~\Cref{sec:case-studies}. 
In \cite{Cormen:2009} the hiring problem is solved through a high-level proof using indicator variables while our formal proof is based on the verified \permute/ functionality and leverages the support of the probability monad for reasoning with expectations in combination with summation theorems.

In~\cite{AvanziniBGMV24} the formalisation of randomised quickselect is established
in~\easycrypt{} exploiting a relational Hoare logic, dedicated to the expressibility
of expected cost analyses of probabilistic data structure. The formalisation effort
is on par to our encoding, cf.~\Cref{table:loc-comp}, although not directly comparable.
Randomised quickselect is implemented in~\cite{AvanziniBGMV24} in a variant of Dijkstra's command language
and the formalisation combined a clever abstraction of the code in lieu with a relational
program logic. Our base implementation is functional and functional correctness
as well a cost analysis is achieved via monadic reasoning.

\section{Conclusion}
\label{Conclusion}

We presented a probability monad for \LH/ that enables the automated verification of expected values and expected costs in probabilistic programs.
Our approach supports discrete distributions with finite support and leverages refinement types to encode probabilistic specifications that are discharged automatically by SMT solvers.
Through several case studies we demonstrated both the automation and expressiveness of our framework.

In future work, we aim to extend our monad to support enumerable discrete distributions with infinite support.
The current restriction to finite support arises from \LH/'s requirement that functions inside refinement types are terminating.
However, it is sound in principle to allow functions that terminate almost surely;
realizing this would require extending \LH/'s theory and termination checker to reason about probabilistic termination with probability one. 

\section*{Data Availability Statement}
\phantomsection\addcontentsline{toc}{section}{Data Availability Statement}

The complete source code of the probability monad together with the case studies is available \cite{artifact}.

\begin{acks}
We would like to thank the Haskell Symposium 2026 reviewers for their valuable comments and suggestions.

The authors wish to acknowledge the support of the FWF project AUTOSARD:
``Automated Sublinear Amortised Resource Analysis of Data Structures''
No.~P36623 and the project VASSAL: ``Verification and Analysis for
Safety and Security of Applications in Life'' funded by the European
Union under Horizon Europe WIDERA Coordination and Support
Action/Grant Agreement No.~101160022 \raisebox{-1pt}{\includegraphics[width=.03\textwidth]{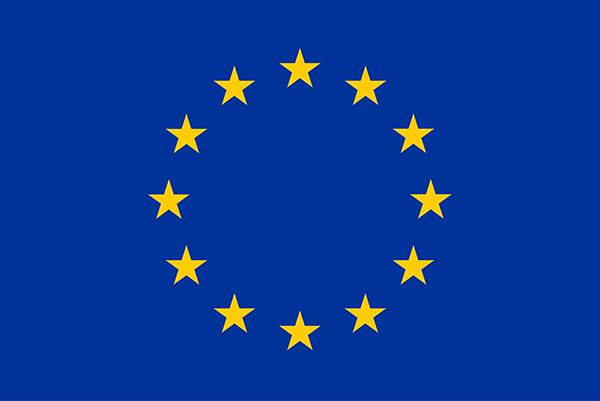}}.
\end{acks}

\appendix

\renewcommand\thefigure{\thesection.\arabic{figure}}
\setcounter{figure}{0}

\section{Appendix}
\label{sec:appendix}

\subsection{Bayes' Theorem}
\label{sec:bayes}

Bayes' theorem states that
\[
\mathbb{P}(A|B) = \frac{\mathbb{P}(B|A)\mathbb{P}(A)}{\mathbb{P}(B)}
\]
where $A, B$ are \emph{events}, i.e., subsets of outcomes of a probability distribution, with $\mathbb{P}(B) \neq 0$, and $\mathbb{P}(A|B)$ denotes the \emph{conditional probability} of event $A$ occurring provided that $B$ has happened.
Conditional probability is defined as
\[
\mathbb{P}(A|B) = \frac{\mathbb{P}(A \cap B)}{\mathbb{P}(B)}
\]

In order to verify Bayes' theorem in terms of the cost-aware probability monad we have to formalise events and conditional probabilities.
An event is modelled using the type \Event/ as a function from values to Booleans, where the function is used to filter the outcomes belonging to the event.
Then, the probability of an event with respect to a distribution is defined as
\begin{lstlisting}[style=haskell]
type ~{\color{constructor}Event}\label{lst:Event}~ #a = (#a -> #B)

{-@ ~\DEF{evProb}~ :: ~\Event/~ #a -> d:~\SubDist/~ #a -> {r:$\Rnneg$ | r <= ~\probMass/~ d} @-}
evProb f ~\DEmpty/~ = 0.0
evProb f (~\DSome/~ (~\Outcome/~ _ v' p) r) = (if f v' then p else 0) + evProb f r
\end{lstlisting}

Then, \C!condProb! defines the conditional probability using the helper function \andF/, where \C!$\andF/$ a b! defines the intersection of events \C!a! and \C!b!.
\begin{lstlisting}[style=haskell]
{-@ ~\DEF{andF}~ :: ~\Event/~ #a -> ~\Event/~ #a -> ~\Event/~ #a @-}
andF f1 f2 x = (f1 x) && (f2 x)

{-@ ~\DEF{condProb}~ :: d:~\SubDist/~ a -> event:~\Event/~ #a -> {given:~\Event/~$\,\alpha\,|\,$~\evProb/~ given d > 0} -> $\Rnneg$ @-}
condProb d event given = ~\evProb/~ (~\andF/~ event given) d / ~\evProb/~ given d
\end{lstlisting}

For the proof of Bayes' theorem we need two helper lemmas.
First, \thmEvProbExt/ establishes extensionality, i.e., \C!$\evProb/$ f d! can be replaced by \C!$\evProb/$ g d! if the events \C!f! and \C!g! agree on all inputs.
\begin{lstlisting}[style=haskell]
{-@ ~\DEF{thm_evProb_ext}~ :: d:~\SubDist/~ a -> f:~\Event/~ a -> g:~\Event/~ a -> (x:a -> { f x = g x }) ->
      { ~\evProb/~ f d = ~\evProb/~ g d } @-}
thm_evProb_ext ~\DEmpty/~ _ _ _ = ()
thm_evProb_ext (~\DSome/~ (~\Outcome/~ c v p) r) f g pr = pr v |> thm_evProb_ext r f g pr
\end{lstlisting}

Moreover, \thmJointProbZero/ shows that if the probability of one event is zero, then also the intersection with any other event results in a probability of zero.

\noindent\begin{minipage}{\linewidth}
\begin{lstlisting}[style=haskell]
{-@ ~\DEF{thm_jointProb_zero}~ :: d:~\SubDist/~ a ->  {f:~\Event/~ #a | ~\evProb/~ f d = 0} -> g:~\Event/~ #a ->
        { ~\evProb/~ (~\andF/~ f g) d = 0 } @-}
thm_jointProb_zero ~\DEmpty/~ _ _ = ()
thm_jointProb_zero (~\DSome/~ (~\Outcome/~ c v p) r) f g | f v       = ~\evProb/~ f r |> unreachable
                                                 | otherwise = thm_jointProb_zero r f g
\end{lstlisting}
\end{minipage}

Note that the first case of the inductive proof of \thmJointProbZero/ makes use of the expression \C!unreachable!, which is provided by \LH/.
Its refinement type is \C!{v:() | false}!, that is, verification can only succeed in contexts with contradictory assumptions.
In our case, if \C!f v! holds this contradicts the assumption that \C!$\evProb/$ f d! equals \C!0! as this can only happen if \C!f! evalutes to \C!False! for all values contained in \C!d!.

Now Bayes' theorem can be easily formalised and verified with ease by a single case distinction:

\noindent\begin{minipage}{\linewidth}
\begin{lstlisting}[style=haskell]
{-@ thm_bayes :: d:~\ProperDist/~ a -> evA:~\Event/~ a -> {evB:~\Event/~ a | evProb evB d > 0} ->
      { ~\condProb/~ d evA evB = (~\condProb/~ d evB evA) * (~\evProb/~ evA d) / (~\evProb/~ evB d) } @-}
thm_bayes d evA evB
    | ~\evProb/~ evA d == 0.0 = ~\thmJointProbZero/~ d evA evB
    | otherwise = ~\thmEvProbExt/~ d (~\andF/~ evA evB) (~\andF/~ evB evA) (\x -> ~\andF/~ evA evB x |> ())
\end{lstlisting}
\end{minipage}

In the first case we use \thmJointProbZero/ to establish that \C!$\condProb/$ d evA evB! is equal to \C!0!.
Otherwise, the theorem can be proven by \LH/ by using \thmEvProbExt/ to obtain the equality between \C!$\evProb/$ ($\andF/$ evA evB)! and \C!$\evProb/$ ($\andF/$ evB evA)!.

\medskip
\verificationEffort{
The formalisation of Bayes' Theorem was straightforward by defining events as predicates over outcomes and establishing the helper theorems \thmJointProbZero/ and \thmEvProbExt/ used in the respective cases in the proof.
Moreover, as declared in~\Cref{table:loc}, a few lines of code suffice for the mechanisation.
}
\subsection{Randomised Quickselect}
\label{sec:qselect}

Quickselect is a search algorithm that finds $k$-th smallest element in a given list~\cite{Hoare61a}.
It operates in a similar fashion as quicksort with the difference that recursive calls are only performed on the part of the input list that contains the searched-for element.

We show functional correctness of randomised quickselect in \LH/ using the cost-aware probability monad and establish that
the expected number of comparisons performed by the algorithm on a list with $n$ distinct elements is at most $4n$.
That is, randomised quickselect achieves a linear time complexity.

\paragraph{Implementation}
The implementation of randomised quickselect is given in~\Cref{fig:quick-select}.
Note that functional correctness is achieved through the annotation in line~\ref{lst:qselect:correct}, i.e., every value in the contained distribution has a rank of \C!k! in the given input list.
Some lemmas are needed in the implementation in order to establish the functional correctness, i.e., \C!body_case_sm!,  \C!body_case_bg!, and \C!castRank!.

The first two lemmas allow exploiting the results in recursive calls that search in a part of the partitioned list for the given list.
Moreover, \C!castRank! is needed to reason when the element removed by \C!deleteAt l i! is added back to the list, since this changes the rank of the element that is searched for.
These lemmas are used in combination with \fMap/ and \LH/ is able to check that the refinement types are correct.
The number of comparisons is again being tracked by \tick/ such that every invocation of \partition/ incurs a cost corresponding to the length of the partitioned list.

\begin{figure}[t]
\begin{lstlisting}[style=haskell,numbers=left]
{-@ ~\DEF{qselect}~ :: Ord #a => {l:~\List/~ #a | noDuplicates l} -> {k:#N | k < length l} ->
      ~\ProperDist/~ {r:#a | k = rank l r} / [length l, 1] @-} ~\label{lst:qselect:correct}~
qselect l k = ~\uniform/~ 0 (length l - 1) (~\qselectBody/~ l k)

{-@ ~\DEF{qselectBody}~ :: Ord #a => {l:~\List/~ #a | noDuplicates l} -> {k:#N | k < length l} ->
      {i:#N | i < length l} -> ~\ProperDist/~ {r:#a | k = rank l r} / [length l, 0] @-}
qselectBody l k i =
    let e = ref l i in
    let (sm, bg) = ~\partition/~ e (deleteAt l i) in
    ~\tick/~ (length l - 1) (
        if k == length sm then
          rank l e |> ~\dunit/~ e ~\label{lst:qselect:k-eq}~
        else if k < length sm then
          smaller l e |> ~\fMap/~ (~\qselect/~ sm k) (body_case_sm l k e) ~\label{lst:qselect:k-lt}~
        else -- k > length sm
             rank l e ~\label{lst:qselect:k-gt}~
           |> ~\fMap/~ (~\qselect/~ bg (k-(length sm)-1))
                ((castRank l (k-1) e) $\circ$ (body_case_bg (deleteAt l i) (k-(length sm)-1) e))
    )

{-@ body_case_sm :: Ord #a => {l:~\List/~ #a | noDuplicates l} -> k:#Z -> e:#a ->
        {r:#a | r < e && k = rank (smaller l e) r} -> {res:#a | res = r && k = rank l r} @-}

{-@ body_case_bg :: Ord #a => {l:~\List/~ #a | noDuplicates l} -> k:#Z ->
      {e:#a | not (e $\in$ l)} -> {r:#a | r > e && k = rank (bigger l e) r} ->
      {res:#a | res = r && k + length (smaller l e) = rank l r} @-}

{-@ castRank :: Ord #a => {l:~\List/~ #a | noDuplicates l} -> k:#Z -> {e:#a | e $\in$ l} ->
      {r:#a | r > e && k = rank (removeFirst e l) r} ->
      {res:#a | res = r && k + 1 = rank l r} @-}
\end{lstlisting}

\caption{Randomised Quickselect}
\label{fig:quick-select}
\end{figure}

\paragraph{Verification}
The proof establishing an upper bound on the expected number of comparisons follows the same approach as our analysis of randomised quicksort.
That is, using the rank-index bijection described in the quicksort formalisation, we first show that \C!$\HEcostTT{\qselect/\ l\ k}$! equals \C!qsel_rec (length l) k!
Where the quickselect recurrence \C!qsel_rec! is defined as
\begin{lstlisting}[style=haskell,xleftmargin=0pt]
qsel_rec n k = n - 1 + (1/n) * ($\HfinSumTT{0}{k-1}$ (\i -> qsel_rec (n-i-1) (k-i-1)) + $\HfinSumTT{k+1}{n-1}$ (\i -> qsel_rec i k))
\end{lstlisting}

Using the summation library we then show the mentioned upper bound by following the steps from \cite{qsel-rec}, resulting in following theorem: \\
\begin{minipage}{\linewidth}\begin{lstlisting}[style=haskell,xleftmargin=0pt]
{-@ thm_qselect_cost_bound :: Ord #a => {l:~\List/~ #a | noDuplicates l} -> {k:#N | k < length l} -> 
        { $\HEcostOpen$~\qselect/~ l k$]$ <= 4 * (length l) } @-}
\end{lstlisting}\end{minipage}

\medskip
\verificationEffort{
Both functional correctness and the expected-cost analysis for randomised quickselect is more involved than the one performed for randomised quicksort in~\Cref{sec:rand-quick}.
While the sorted sub-lists can be easily combined into a sorted list with the help of refinement types in the case of quicksort, the reasoning about the rank of the searched-for element in quickselect required three helper theorems.
Additionally, the extracted recurrence is more complicated and only an upper bound can be shown.
}
\subsection{Logarithmic Proofs}
\label{appendix:log}

A useful logarithmic inequality that we employ in two of the case studies is the following:
\[ 2 + \log_2(x) + \log_2(y) \leqslant 2\log_2(x+y) \quad \text{ for } x, y \geqslant 1. \]
This inequality has also been used by Okasaki \cite{Okasaki:1999}.
In the formalised proof, presented below, we have to instantiate the needed assumptions one after another such that the SMT solver is able to verify the inequality from the given facts.

First, we establish that \C!$\logTwo/$ 1 = 0! by using \C!$\logTwoProduct/$ 2 1!, which results in the fact that \C!$\logTwo/$ 2! is equal to \C!$\logTwo/$ 2 + $\logTwo/$ 1!.
The equality can be transformed into \C!1 = 1 + $\logTwo/$ 1! by using \logTwoBase/, from which \C!$\logTwo/$ 1 = 0! is evident.
\begin{lstlisting}[style=haskell]
{-@ thm_log~\textsubscript{2}~_1_eq_0 :: {~\logTwo/~ 1 = 0} @-}
thm_log~\textsubscript{2}~_1_eq_0 = ~\logTwoProduct/~ 2 1 |> ~\logTwoBase/~
\end{lstlisting}

In the following proof of \C!thm_log$\textsubscript{2}$_inequality! we state in comments for each used axiom how the information is used to gradually build the proof.
\begin{lstlisting}[style=haskell]
{-@ thm_log~\textsubscript{2}~_inequality :: x:$\Zpos$ -> y:$\Zpos$ -> {2 + ~\logTwo/~ x + ~\logTwo/~ y <= 2 * ~\logTwo/~ (x+y)} @-}
thm_log~\textsubscript{2}~_inequality x y =
       ~\logTwoProduct/~ (x+y) (x+y)            -- $\color{commentGray}\log_2((x+y)\cdot (x+y)) = 2\log_2(x+y)$
    |> ~\logTwoMonotone/~ (4*x*y) ((x+y)*(x+y)) -- $\color{commentGray}\log_2(4xy) \leqslant 2\log_2(x+y)$
    |> ~\logTwoProduct/~ (4*x) y                -- $\color{commentGray}\log_2(4x) + \log_2(y) \leqslant 2\log_2(x+y)$
    |> ~\logTwoProduct/~ 4 x                    -- $\color{commentGray}\log_2(4) + \log_2(x) + \log_2(y) \leqslant 2\log_2(x+y)$
    |> ~\logTwoProduct/~ 2 2 |> thm_log~\textsubscript{2}~_1_eq_0 -- $\color{commentGray}2 + \log_2(x) + \log_2(y) \leqslant 2\log_2(x+y)$
\end{lstlisting}
In the call to \logTwoMonotone/ the inequality $4xy \leqslant (x+y)^2$ is automatically validated by the SMT solver.
In a pen-and-paper proof (cf. \cite[Lemma 16]{Hofmann22}) one can argue that $(x+y)^2 - 4xy = (x-y)^2 \geqslant 0$.

\subsection{Summation Theorems}

A list of summation theorems provided by our library and not shown in~\Cref{sec:summations} is given in \Cref{fig:finSum-theorems}.

\begin{figure}[H]
\begin{minipage}[t]{0.48\linewidth}%
\begin{lstlisting}[style=haskell,xleftmargin=0pt,aboveskip=0pt,belowskip=0pt]
~\DEF{thm_sum_factor}~ :: lo:#Z -> hi:#Z ->
    c:#R -> f:(~\Range/~ lo hi -> #R) ->
    g:(i:~\Range/~ lo hi -> {v:#R | c * v = f i}) ->
    { $\HfinSumTT{lo}{hi}$ f = c * $\HfinSumTT{lo}{hi}$ g }

~\DEF{thm_sum_constant}~ :: lo:#Z ->
    {hi:#Z | lo <= hi} -> c:#R ->
    f:(i:~\Range/~ lo hi -> {v:#R | v = c}) ->
    { $\HfinSumTT{lo}{hi}$ f = c * (hi - lo + 1) }

~\DEF{thm_sum_rewrite}~ :: lo:#Z -> hi:#Z ->
    f:(~\Range/~ lo hi -> #R) ->
    g:(i:~\Range/~ lo hi -> {v:#R | v = f i}) ->
    { $\HfinSumTT{lo}{hi}$ f = $\HfinSumTT{lo}{hi}$ g }
\end{lstlisting}
\end{minipage}%
\hfill\vrule width 0.6pt\hfill
\begin{minipage}[t]{0.48\linewidth}%
\begin{lstlisting}[style=haskell,xleftmargin=0pt,aboveskip=0pt,belowskip=0pt]
~\DEF{thm_sum_linear}~ :: lo:#Z -> hi:#Z ->
    left:(~\Range/~ lo hi -> #R) ->
    right:(~\Range/~ lo hi -> #R) -> 
    f:(i:~\Range/~ lo hi ->
        {v:#R | left i + right i}) ->
    { $\HfinSumTT{lo}{hi}$ f = $\HfinSumTT{lo}{hi}$ left + $\HfinSumTT{lo}{hi}$ right }

~\DEF{thm_sum_reverse}~ :: lo:#Z -> hi:#Z ->
    g:(~\Range/~ lo hi -> #R) -> 
    f:(~\Range/~ lo hi -> #R) ->
    (i:~\Range/~ lo hi -> { g i = f (hi+lo-i)}) ->
    { $\HfinSumTT{lo}{hi}$ g = $\HfinSumTT{lo}{hi}$ f }

~\DEF{thm_sum_shift}~ :: lo:#Z -> hi:#Z -> k:#N ->
    g:(~\Range/~ (lo+k) (hi+k) -> #R) ->
    f:(i:~\Range/~ lo hi -> {v:#R | v = g (i+k)}) ->
    { $\HfinSumTT{lo}{hi}$ f = $\HfinSumTT{lo+k}{hi+k}$ g }
\end{lstlisting}
\end{minipage}

\caption{Finite Sum Theorems}
\label{fig:finSum-theorems}
\end{figure}

\bibliographystyle{ACM-Reference-Format}
\bibliography{bib}

\end{document}